\documentclass[letterpaper,11pt]{article}
\usepackage[hyphens]{url}
\usepackage{graphicx}
\usepackage[margin=1in]{geometry}
\usepackage{hyperref}
\usepackage[numbers]{natbib}
\usepackage{caption}
\usepackage{algorithm}
\DeclareCaptionStyle{ruled}{labelfont=normalfont,labelsep=colon,strut=off}
\floatstyle{ruled}
\usepackage{amsfonts,amsmath,amsthm,amssymb}
\usepackage{thm-restate}
\usepackage{algpseudocode}
\usepackage{enumitem}
\usepackage{accents}
\usepackage{lipsum}

\newtheorem{theorem}{Theorem}
\newtheorem{lemma}[theorem]{Lemma}

\newtheorem{definition}[theorem]{Definition}
\newtheorem*{remark}{Remark}

\newcommand{\R}{\mathbb{R}}

\newcommand{\sL}{\mathcal{L}}

\newcommand{\sD}{\mathcal{D}}

\usepackage{tikz}
\usetikzlibrary{shapes.geometric, arrows}
\tikzstyle{basicbox} = [rectangle, rounded corners, minimum width=1cm, minimum height=1cm,text centered, draw=black, fill=white!30]
\tikzstyle{fakeNode} = [rectangle, rounded corners, minimum width=1cm, minimum height=1cm,text centered, draw=white, fill=white!30]
\tikzstyle{arrowDet} = [thick,->,>=stealth]
\tikzstyle{arrowProb} = [thick,dash pattern=on 2pt off 2pt,->,>=stealth]

\usepackage[suppress]{color-edits} 
\usepackage[dvipsnames]{xcolor}
\addauthor{JG}{red!90!black}

\addauthor{VN}{blue!90!black}

\addauthor{YS}{green!80!black}

\addauthor{TODO}{brown}

\author{Jugal Garg\thanks{University of Illinois Urbana-Champaign}\\ \texttt{\small jugal@illinois.edu}
\and
Vishnu V. Narayan\footnotemark[2] \\ \texttt{\small vishnu.narayan@mail.mcgill.ca}
\and 
Yuang Eric Shen\footnotemark[2]\\ \texttt{\small yuangs3@illinois.edu
}}
\date{}
\title{Designing Truthful Mechanisms for Asymptotic Fair Division\thanks{Work supported by NSF Grant CCF-2334461.}}

\begin{document}

\maketitle

\begin{abstract}
We study the problem of fairly allocating a set of $m$ goods among $n$ agents in the asymptotic setting, where each item's value for each agent is drawn from an underlying joint distribution. Prior works have shown that if this distribution is well-behaved, then an envy-free allocation exists with high probability when $m=\Omega(n\log{n})$ (\citet{dickerson2014}). Under the stronger assumption that item values are independently and identically distributed (i.i.d.) across agents, this requirement improves to $m=\Omega(n\log{n}/\log{\log{n}})$, which is tight (\citet{manurangsi2021closing}). However, these results rely on non-strategyproof mechanisms, such as maximum-welfare allocation or the round-robin algorithm, limiting their applicability in settings with strategic agents. 

In this work, we extend the theory to a broader, more realistic class of joint value distributions, allowing for correlations among agents, atomicity, and unequal probabilities of having the highest value for an item. We show that envy-free allocations continue to exist with a high probability when $m=\Omega(n\log{n})$. More importantly, we give a new randomized mechanism that is truthful in expectation, efficiently implementable in polynomial time, and outputs envy-free allocations with high probability, answering an open question posed by \citet{manurangsi2017asymptotic}. We further extend our mechanism to settings with asymptotic weighted fair division and multiple agent types and good types, proving new results in each case.
\end{abstract}

\section{Introduction} \label{sec:intro}

The question of how to fairly divide a collection of items among a group of agents prominently appears in societal contexts. It arises in settings such as the division of inherited estates, the allocation of computational resources among competing tasks, and the assignment of limited seats in university courses among students. In each case, fairness in treatment is a fundamental concern. This question, formalized in the literature as the fair division problem, has received significant recent attention (see, e.g., the survey of \citet{amanatidis2023fair}).

Among the various notions of fairness, \textit{envy-freeness} plays a central role. An allocation is said to be envy-free (EF) if no agent prefers another agent's bundle over their own. Envy-free allocations are known to exist under mild assumptions when items can be divided continuously (\citet{dubins1961cut,stromquist1980cut,brams1995envy}). However, in the indivisible setting, in which each item must be allocated to exactly one agent, envy-freeness is not always achievable. For example, when one valuable item must be allocated between two agents, any allocation inevitably causes envy in the agent who does not receive it.

This raises a fundamental question: When does an envy-free allocation exist? On the one hand, a simple reduction from the \emph{partition} problem shows that even when agents have identical additive valuation functions, deciding whether a given instance admits an envy-free allocation is NP-hard. However, envy-freeness appears to be quite prevalent in real-world instances. For example, on the popular nonprofit platform \texttt{spliddit.org} (\citet{goldman2015spliddit}), \citet{bai2022fair} reports that over 70\% of the submitted instances admit an EF allocation. Moreover, non-existence of such allocations seems primarily to arise in small instances; in cases where the number of items is at least three times the number of agents, 93\% of instances have an EF allocation. It is therefore important to gain an understanding of the structure of instances that admit an EF allocation, and identify conditions under which such allocations can be found efficiently. 

Intuitively, following the discussion above, one might expect that instances with many high-valued items and diverse agent valuations are likely to admit EF allocations. Motivated by this, a major focus in the literature has been on identifying sufficient conditions under which EF allocations are guaranteed to exist. In this direction, \citet{dickerson2014} initiated the study of asymptotic fair division. They considered settings with additive valuations, in which each item's value is independently drawn from a `well-behaved' joint distribution across all agents, and where each agent has equal probability of having the highest value for each item. They showed that for any number of agents $n$, if the number of items is large relative to $n$ (specifically, $m=\Omega(n\log{n})$), then with high probability, as $m \rightarrow \infty$, an EF allocation exists. Moreover, such an allocation can be obtained by a simple greedy algorithm: assign each item to an agent who values it most. In subsequent work, \citet{manurangsi2020envy} showed in a similar asymptotic setting that if $n$ does not divide $m$, then even $m=\Theta(n\log{n}/\log\log{n})$ items may be insufficient for envy-freeness. This gap was later closed by \citet{manurangsi2021closing}, who showed that, under stronger distributional assumptions, $m=\Omega(n\log{n}/\log\log{n})$ items suffice when the allocation is made using the round-robin procedure.

Although these results significantly advance our understanding of the above questions, they come with important limitations. First, while both welfare maximization and the round-robin mechanism are simple and efficient, they lack a key property in mechanism design: \emph{truthfulness}. A mechanism is truthful (or strategyproof) if no agent can benefit by misreporting their valuation. In fact,  \citet{manurangsi2017asymptotic} posed the existence of a truthful mechanism that ensures envy-freeness in the asymptotic setting as an open problem. Second, the assumptions made on the underlying distributions in previous works are rather restrictive. For instance, \citet{dickerson2014} assumes that each agent has an equal probability ($1/n$) of being the highest valuing agent for each item. And,  \citet{manurangsi2021closing} requires that item values are independent and identically distributed across agents, with a density bounded between constants. Finally, while these works examine the behavior of `typical' instances drawn from a well-behaved distribution, they offer little insight into the structure of a specific instance once realized. This raises another natural question: given a particular instance, can we efficiently verify whether it contains sufficiently many high-valued items and sufficiently diverse agent valuations to ensure the existence of an EF allocation?

\subsection{Our Results} \label{subsec:results}
In this work, our main result is a new randomized mechanism, the Proportional Response with Dummy (PRD) Mechanism, that addresses these concerns.

\textbf{Main Result (informal).}
    \textit{The PRD Mechanism for asymptotic fair division is truthful in expectation, runs in polynomial time, and outputs an envy-free allocation in typical asymptotic instances (i.e., with high probability), provided that $m=\Omega(n\log{n})$.}

Our contributions in this paper have several novel aspects.
\begin{itemize}
    \item Our main result, the new PRD Mechanism, improves upon the state of the art in two key ways. First, we impose very mild assumptions on the underlying distribution of item values. Our asymptotic setting allows for correlation between agent values and is strictly more general than that of \citet{manurangsi2021closing} and the related work of \citet{bai2022envy}. It also allows for many broad properties that \citet{dickerson2014} excluded, such as atomicity and unequal probabilities of having the highest value for an item. Secondly, unlike the mechanisms employed in these works, the PRD Mechanism is \textit{truthful in expectation}: no agent can obtain a bundle of greater expected value by misreporting its valuation function. This is a surprising and rare positive result in light of the large collection of negative results on the impossibility of truthfulness in many fair division settings (e.g.~\citet{lipton2004approximately,caragiannis2009low,amanatidis2017truthful}). Furthermore, the PRD Mechanism is polynomial-time implementable. 
    \item We introduce the use of the \textit{Kullback–Leibler (KL) divergence}, a statistical measure of the distinctness between two distributions, in the context of asymptotic fair division, connecting the KL divergence between appropriately normalized valuations to the envy-margin between the agents. Informally, we show that, after rounding up values very close to 0 to a common lower bound, if the resulting KL divergences between each pair of normalized valuations is high, then there exists a fractional allocation with a high envy margin for all pairs of agents.
    \item Finally, we study the implications of our results for other fair division settings. We show that our results extend to the setting of asymptotic weighted fair division, in which agents have varying entitlements, simultaneously generalizing previous existence results while maintaining truthfulness in that setting. One consequence of our result is that it extends to cases in which agents may have weights that grow with $n$. For instance, even if a government agency or union is entitled to a constant fraction of the total utility (i.e., has weight linear in $n$), our result implies $\Theta(n\log(n))$ items can suffice. We also present a truthful mechanism for asymptotic fair division for \textit{groups} with weights, introduced in \citet{manurangsi2017asymptotic}. Finally, we consider multiple \textit{agent types} and \textit{good types}, introduced in \citet{Gorantla_2023}, and provide interesting new results.
\end{itemize}

We remark that deciding whether an envy-free allocation exists is computationally NP-hard, so no polynomial-time mechanism can always find an envy-free allocation whenever one exists. Similarly, no deterministic truthful mechanism can always find an envy-free allocation when one exists (\citet{lipton2004approximately}).

\subsection{Additional Related Work} \label{sec:related}

The asymptotic fair division problem was introduced by \citet{dickerson2014}, who showed that EF allocations exist with high probability when $m = \Omega(n \log{n})$. The existence of envy-free allocations in the asymptotic setting for goods has also been considered in a series of works by \citet{manurangsi2017asymptotic,manurangsi2020envy,manurangsi2021closing} and by \citet{bai2022envy}. Recently, \citet{manurangsi2025chores} extended this line of research to the division of chores, and \citet{manurangsi2025weighted} considered the setting where the agents have weights.

Other fairness notions, such as the maximin share, proportional share, and approximations thereof, have also been considered in the asymptotic setting (see, e.g., the works of \citet{amanatidis2015approximation,kurokawa2016can,amanatidis2016truthful,suksompong2016asymptotic,farhadi2019fair}). In other work, \citet{yokoyama2025asymptotic} study the asymptotic existence of what they call `class envy-free matchings', \citet{benade2024fair} study asymptotic fair division in the online setting, \citet{bai2022fair} investigate a `smoothed' utility model with perturbed valuations, and \citet{benade2024existence} study a related stochastic setting in which item labels are randomly shuffled for each agent.

Truthful mechanisms for fair division have also been considered in the non-asymptotic setting. In a seminal work, \citet{lipton2004approximately} showed that no truthful mechanism outputs an envy-free allocation whenever one exists. Similar impossibility results were published in \citet{caragiannis2009low} and \citet{amanatidis2017truthful}. A recent paper of \citet{bu2025truthful} studies the existence of mechanisms that are truthful in expectation and that output allocations that are approximately envy-free.

\section{Preliminaries} \label{sec:preliminaries}

An instance of the fair division problem consists of a set $N = \{1, \ldots, n\}$ of agents represented by indices in $[n]$, a set $M = \{1, \ldots, m\}$ of indivisible items represented by indices in $[m]$, and a valuation profile $(v_i)_{i \in N}$. We assume throughout that the valuation functions are additive, i.e., that the value of a set of items is the sum of the items' singleton values. Our main focus is on the case of goods, so we represent agent~$i$'s valuation as a vector ${v_i}$ that assigns a non-negative real value $v_{ij}$ for each item $j$. The total value $v_i^T S$ for any bundle of items $S\in\{0,1\}^m$ is equal to $\sum_{j: S_j=1} v_{ij}$. An integral allocation $A = (A_1, \ldots, A_n)$ of items to agents is a partition of the items into $n$ disjoint bundles, where agent $i$ gets bundle $A_i \in \{0,1\}^M$ and obtains value $v_i^TA_i$. While our main results concern integral allocations, our analysis also includes \textit{fractional} allocations in which each item may be assigned to multiple agents in fractional amounts. We say agent $i$ has a fractional allocation $x_i \in [0,1]^m$, where $x_{ij} \in [0,1]$ is the fraction of good $j$ assigned to agent $i$, and the value of this allocation is $v_ix_i$. We assume all allocations are complete, i.e., $\sum_{i \in N} x_{ij} = 1$ for each item $j$ and no item is left (partially) unassigned.

For the analysis of our mechanism, we often use \textit{normalized} valuations, which are scaled for each agent so that the sum of their item values equals 1. We represent these normalized valuations by $\bar v$, i.e. $\bar v_{ij} = v_{ij} / \sum_{j\in M}v_{ij}$. We say an allocation $A$ is {\em envy-free} if
\[ \bar v_i^T A_i \geq \bar v_i^T A_k \quad \forall i,k \in N.\]
We also define the {\em envy margin} of agent $i$ with respect to agent $k$, for a given allocation $A$, as
\[EM_{ik} = \bar  v_i^T A_i -\bar  v_i^T A_k.\]
Observe that an agent has no envy if its envy margin with respect to every other agent is non-negative, and that an allocation is envy-free if this is true for all agents.  For a given fractional allocation $x$, we denote by $fEM_{ik}$ the fractional envy margin of agent $i$ with respect to agent $k$, which for our analysis is computed with the normalized valuations as $\sum_j \left[\bar v_{ij}^T  x_{ij} - \bar v_{ij}^T x_{kj}\right]$. Additionally, we say that an event occurs \textit{with high probability} if it occurs with probability at least $1 - o(1)$ in $m$, i.e., almost surely as $m$ goes to infinity.

\subsection{Asymptotic Fair Division}
The asymptotic fair division problem aims to understand the conditions under which a fair allocation exists asymptotically, i.e., with high probability, as the number of items grows. In this setting, for every item $j\in M$, the values $v_{1j}, \ldots, v_{nj}$ are drawn from a joint distribution $\sD$ over $[0,1]^n$. Here, we allow correlation among agents' values for a given item, but these values are independent across different items. We also denote the marginal distribution for each agent by $\sD_i$, and the mean of $\sD_i$ by $\mu_i$.

We impose two minor and intuitive assumptions on $\sD$. The first is that there is some constant $\mu_l > 0$ such that for all agents $i$, $\mu_i \geq \mu_l$. In other words, we do not allow the marginal distribution means to approach $0$ as more agents are added to an instance. Another way of viewing this assumption is that the most valuable good for an agent is not arbitrarily larger than the average good. The second assumption is that the \textit{expected absolute difference} between any two agents' values for a given item, normalized by their means, is at least a positive constant. Specifically, we require that there exists some $\delta \in (0,1)$ such that

\[\forall i \neq k, E\bigg [ \bigg | \frac{v_i}{\mu_i}  -\frac{v_k}{\mu_k}\bigg | \bigg ] \geq \delta.\]

This assumption is readily satisfied when agent valuations are i.i.d. and distributions are not fully concentrated at their means. More generally, it is true when agents have, on average, at least constant disagreement on the value of each good, which intuitively is important for the existence of an envy-free allocation.

\subsubsection{KL Divergence.} Given two discrete distributions $P$ and $Q$ over a discrete set $S$ that each sum to 1, the Kullback–Leibler (KL) divergence, or relative entropy, is a widely-used statistical measure of the difference between the two distributions. It is defined as
\[D_{KL}(P\|Q) = \sum_{s\in S} P(s)\cdot(\log(P(s)) - \log(Q(s))).\]

\subsubsection{Truthfulness in Expectation.} A mechanism is truthful (or strategyproof) if no agent can benefit by misreporting its true valuation function. Formally, a mechanism with allocation rule $\mathcal{A}$ is truthful in expectation if, for every bidder $i$, true valuation function $v_i$, reported valuation function $\hat v_i$, and reported valuation functions $\hat v_{-i}$ of the other agents,
\[\mathbb{E}[v_i(\mathcal{A}(v_i,\hat v_{-i}))] \geq \mathbb{E}[v_i(\mathcal{A}(\hat v_i, \hat v_{-i}))]\]
where the expectation is over the mechanism's randomness.

\section{Well-Behaved Distributions} \label{sec:well-behaved}
The early work of \citet{dickerson2014} showed that if item values are drawn from a well-behaved joint distribution independently for each item, then an envy-free allocation exists with high probability as $m\rightarrow\infty$ when $m=\Omega(n\log{n})$. Informally, \citet{dickerson2014} considers a distribution well-behaved if it is non-atomic, each agent has the same probability (i.e. $1/n$) of having the highest value for each item, and the expected value of an item conditioned on the event that an agent has the highest value for it is strictly larger than its expected value otherwise. In a subsequent paper, \citet{manurangsi2021closing} considered the case where the value of each item is additionally independent and identically distributed for every agent, and showed an improvement to the previous upper bound to $m=\Omega(n\log{n}/\log\log{n})$. This result also assumes that the density function of the distribution underlying each item's value is $(\alpha,\beta)$-bounded, i.e., bounded between two constants $\alpha,\beta$. Follow-up work by \citet{bai2022envy} slightly generalizes the assumptions of \citet{manurangsi2021closing} to the case where agents are non-identical, and their $(\alpha,\beta)$-bounded distributions may be supported on a different sub-interval of $[0,1]$ for each agent.

In this work, we allow for atomicity of this distribution and correlation across agents. Our only assumptions, as stated earlier, are that $\forall i \in N$, $\mu_i \geq \mu_l$, and that there exists some $\delta>0$ such that $\forall i \neq k$, $E[|\frac{v_i}{\mu_i}-\frac{v_k}{\mu_k}|] \geq \delta$. These assumptions are strictly more general than those of \citet{bai2022envy} and therefore \citet{manurangsi2021closing}. To see this, for our first assumption, $\mu_i \geq \mu_l > 0$, we note that the upper bound $\beta$ implies a minimum interval length of $\frac{1}{\beta}$, over which the PDF $f_{\sD_i}$ is $\geq \alpha$, therefore $\mu_l = \frac{\alpha}{2\beta}$ suffices. For our second assumption, $\mu_i,\mu_k \leq 1$ implies that the PDFs of $v_i/\mu_i$ and $v_k/\mu_k$ are both bounded above by $\beta$ over their domain, which implies that $\big | \frac{v_i}{\mu_i} - \frac{v_k}{\mu_k} \big | \geq \frac{1}{4\beta}$ with probability at least $\frac{1}{2}$, and thus $\delta = \frac{1}{8\beta}$ suffices to satisfy the second assumption.  

\section{Technical Overview of the PRD Mechanism} \label{sec:overview}
In the following sections we present our mechanism, the PRD Mechanism, and its analysis. The PRD Mechanism operates in two phases. In the first phase, it collects reports from the agents that represent their valuation functions, and internally constructs a profile of `bids' for every agent-item pair. The mechanism then uses these bids to construct a fractional allocation, which has a high envy margin for typical instances (i.e., with high probability). Notably, we will also show that the assigned bids are agent-optimal, so no agent can increase the value of their fractional bundle by misreporting. In the second phase, the PRD Mechanism employs a randomized rounding scheme to round this allocation and obtain an integral final allocation. This rounding step ensures that the expected value of the final allocation for each agent equals its value for the fractional allocation, maintaining truthfulness in expectation. Critically, the mechanism guarantees that for typical instances the fractional envy margin prior to the rounding step is high enough that the resulting integral allocation is envy-free with high probability.

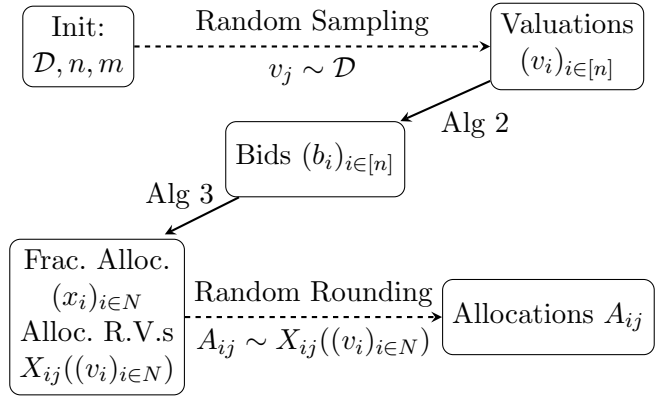
\begin{figure}[h]
    \centering
    \begin{tikzpicture}
    \node (Init) [basicbox, align=center] {Init:\\$\mathcal{D}, n, m$};
    \node (Vi) [basicbox, right of=Init, node distance = 6.5cm, align=center] {Valuations\\$(v_i)_{i \in [n]} $};
    \node (FakeNode1) [fakeNode,  below of=Init, node distance = 1.5cm] {};
    \node (FakeNode2) [fakeNode,  right of=FakeNode1, node distance = 0.25cm] {};
    \node (Bids) [basicbox, right of=FakeNode2, node distance = 2.9cm] {Bids $(b_i)_{i \in [n]} $};
    \node (FracAlloc) [basicbox, below of=FakeNode2, node distance = 2.1cm, align=center] {Frac. Alloc.\\$(x_i)_{i \in N}$\\Alloc. R.V.s\\$X_{ij}((v_i)_{i \in N})$} ;
    \node (Alloc) [basicbox, right of=FracAlloc, node distance = 6cm ] {Allocations $A_{ij}$};
    \draw [arrowProb] (Init) --  node[anchor=south] {Random Sampling}  node[anchor=north] {$v_j \sim \mathcal{D}$}  (Vi);
    \draw [arrowDet] (Vi) --  node[anchor=north west,shift={(-2mm,0mm)}] {Alg 2}  (Bids);
    \draw [arrowDet] (Bids) -- ++ (-2,-1)  node[anchor=south,shift={(2mm,2mm)}] {Alg 3}  (FracAlloc);
    \draw [arrowProb] (FracAlloc) --  node[anchor=south] {Random Rounding}  node[anchor=north] {$A_{ij} \sim X_{ij}((v_i)_{i \in N})$}  (Alloc);
    \end{tikzpicture}
    \caption{Asymptotic Setting and Mechanism Overview. Dashed lines represent randomization, and solid lines represent deterministic algorithms.}
    \label{fig:enter-label}
\end{figure}

The following two subsections present a high-level overview of the ideas used in the development of the PRD Mechanism. The formal mechanism and its analysis are presented in Section~\ref{sec:main}, and several extensions of this mechanism appear in Section~\ref{sec:extensions}.

\subsection{Truthfulness: The Dummy Agent} \label{subsec:truthful}

Our goal for this section is to introduce, in an intuitive manner, the design of the first phase: a truthful mechanism that finds a fractional allocation with large envy margin (that will later be rounded to obtain an integral allocation).

Our analysis of the PRD Mechanism conceptually assigns equal `budgets' to the agents, which they effectively distribute across the items as their `bids'. However, since we desire truthful reporting of the agents' valuations, each agent $i$ will report values $v_{ij} \in [0,1]$ for every item $j\in[m]$, and the mechanism will internally translate these values into bids $b_{ij}$. Reported values must be within the domain of $\sD$ (i.e., in $[0,1]$), but have no sum constraint. The bids constructed by the mechanism, however, are constrained to sum to 1 for each agent. We note that even when values $v_{ij}$ are generated independently across the items for agent $i$, the normalized values $\bar v_{ij}$ are not independent. How should the mechanism construct these bids in order to preserve truthful reporting?

For a first attempt, suppose the agents receive an allocation of each item proportional to their constructed bid for that item, i.e., agent~$i$ receives $x_i$ with $x_{ij} =  b_{ij}/\sum_{k\in[n]}  b_{kj}$ for each good $j$. An example demonstrating this allocation function for an instance with two agents and two items is shown in Table~\ref{tab:bidsallocs}.

\begin{table}[h]
    \centering
    \begin{tabular}{c|c|c}
         & \multicolumn{2}{c}{Bids}  \\
         \hline 
         Agent~1 & 0.2 & 0.8 \\ 
         Agent~2 & 0.4 & 0.6 \\ 
    \end{tabular}
    $\implies$
    \begin{tabular}{c|c|c}
         & \multicolumn{2}{c}{Allotment}  \\
         \hline 
         Agent~1 & 1/3 & 4/7 \\ 
         Agent~2 & 2/3 & 3/7 \\ 
    \end{tabular}
    
    \caption{Proportional allotment}
    \label{tab:bidsallocs}
\end{table}

If we take any solution in which agent~$i$ does not want to deviate, then we necessarily have for any items $j,j'$ for which $x_{ij}>0$ and $x_{ij'}>0$ the first order optimality condition, i.e. that $\bar v_j \cdot \frac{\partial}{\partial \bar b_{ij}} x_{ij} = \bar v_{j'} \cdot \frac{\partial}{\partial \bar b_{ij'}} x_{ij'}$. However, strategic interactions in this game make optimal bids difficult to analyze. To address this, we will use the fact that when given a partial allocation with a certain envy margin, allocating the remaining portion of all goods by dividing it equally across the agents does not change the envy margin. Following from this observation, we can introduce a dummy agent into the above system whose bid is always $n - \sum_i { b_{ij}}$ for each good $j$. Now the items are allocated using the same allocation function as before, but after this step, the dummy agent's share is divided uniformly among the other $n$ agents. Introducing the dummy agent to the previous example yields the result in Table~\ref{tab:bidsallocsdummy}.

\begin{table*}[h]
    \centering
    \begin{tabular}{c|c|c}
         & \multicolumn{2}{c}{Bids}  \\
         \hline 
         Agent~1 & 0.2 & 0.8 \\ 
         Agent~2 & 0.4 & 0.6 \\ 
         Dummy & 1.4 & 0.6 \\ 
    \end{tabular}
    $\implies$
    \begin{tabular}{c|c|c}
         & \multicolumn{2}{c}{Interim Allotment}  \\
         \hline 
         Agent~1 & \hspace*{3mm} 0.1 \hspace*{3mm} & 0.4 \\ 
         Agent~2 & 0.2 & 0.3 \\ 
         Dummy & 0.7 & 0.3 \\ 
    \end{tabular}
    $\implies$
    \begin{tabular}{c|c|c}
         & \multicolumn{2}{c}{Final Allotment}  \\
         \hline 
         Agent~1 & \hspace*{1.25mm} 0.45 \hspace*{1.25mm} & 0.55 \\ 
         Agent~2 & 0.55 & 0.45 \\ 
    \end{tabular}
    \caption{Proportional allotment with dummy agent}
    \label{tab:bidsallocsdummy}
\end{table*}

Observe that because the dummy agent ensures the sum of bids equals $n$ (in our example $n=2$), the allotment at the intermediate stage is exactly the bid divided by $n$. In the final allocation, we can see that with bids $b_{ij}$, agent $i$ is assigned

\[x_{ij} = \frac{ b_{ij}}{n} + \frac{1}{n}\cdot\frac{n-\sum_i  b_{ij}}{n}.\]

With this allocation, $\bar v_{ij} \cdot \frac{\partial}{\partial \bar b_{ij}} x_{ij} = \bar v_{ij}/n - \bar v_{ij}/n^2$ at any bid profile, therefore an agent does not care about the other agents' bids when it maximizes its own value. This means that the introduction of the dummy agent makes this game separable across agents, and each agent will optimize separately in its own equivalent single-agent game. Additionally, the intermediate and final fractional allocations have the same envy margin for any pair of agents, since the allotment of each item changes by the same quantity for every agent after the dummy agent's share is divided among the other agents.

How would an agent want to bid in the above game? When bids translate linearly to allotments, agents will want to bid their entire allowance on the items they value the most. This means agents may disagree on the value of most items, and yet, if they agree on their best items, we do not obtain any envy margin in the fractional allocation. We therefore require a mechanism that ensures that substantially different valuations always map to substantially different fractional allocations. Achieving this requires a few modifications to this mechanism. First, we introduce a function $f: [0,1] \rightarrow \R^+$, and modify the allocation function such that agents receive an allotment of good $j$ proportional to $f(b_i)$ instead of $b_i$. The dummy agent will similarly always bid $n\cdot f(1) - \sum_i {f(b_{ij})}$ on good $j$, which continues to have the effect of making the game separable across agents. Specifically, for agent $i$, for some function $F$ we have

\[x_{ij} = \frac{n-1}{n^2 f(1)} f(b_{ij}) + F_j(b_{-i}).\]

Suppose for now that the bids are equal to the normalized valuations $\bar v := \frac{v_{ij}}{\sum_j v_{ij}}$. When $f$ induces truthful reporting, distinct values naturally correspond to distinct bids. Given normalized values $\bar v_{ij}$ and $\bar v_{ij'}$, we require

\[\bar v_{ij} f'(\bar v_{ij}) = \bar v_{ij'} f'(\bar v_{ij'})\]

whenever $f'$ exists for both values. The simplest way to achieve this is to require $f'(\bar v_{ij}) = \frac{1}{\bar v_{ij}}$, which gives us $f(b_{ij}) = \log(b_{ij}) + c$, the injective function we will use. 

At this stage, an obvious problem is that $\log(b_{ij}) + c$ may be negative. There is also a less obvious problem, which is that the dummy agent will have to bid a huge amount to cover all scenarios, since $\bar v_{ij}$ is typically about $\approx \Theta(\frac{1}{m})$ but can be as high as 1 in the worst case. We solve these problems by imposing a floor and ceiling on constructed bids. Specifically, we demand that each bid $b_{ij}$ is restricted to some interval $[b_{min}, b_{max}]$ that we will specify later. We define the projection $proj(\bar v_{ij}) = min\{max\{\bar v_{ij}, b_{min}\}, b_{max}\}$ of $\bar v_{ij}$ into the interval $[b_{min}, b_{max}]$. Then, in order to ensure that the sum of $proj(\bar v_{ij})$ values matches the budget constraint, we will multiply $\bar v_i$ by some \textit{scale factor} $s_i$. We will show that for our choice of interval, with high probability there exists some $s_i > 0$ that ensures $\sum_j proj(s_i\bar v_{ij}) = 1$, giving us our optimal bids. Importantly, for typical instances applying this scaling and projection step respects the optimality conditions for each agent, maintaining the optimality of truthful reporting.

To ensure that constructed bids are nonnegative, we choose $c = -\log (b_{min})$. We will also define $C := f(b_{max})$, and design the dummy agent such that the agents now receive a proportional allotment out of a total of $nC$. Then, the allocation to agent $i$ becomes
\[x_{ij} = \frac{\log(b_{ij}) + c}{nC} + \frac{nC - \sum_k \big (\log( b_{kj}) + c \big )}{n^2C}.\]

\subsection{Envy-freeness: Envy Margin via KL Divergence} \label{subsec:envymargin}

In this section, we explain the connection between the envy margin obtained by our mechanism  at the end of the first phase and the KL divergence between the agents' valuations. Specifically, if the KL divergence is high, then the resulting fractional allocation has a high envy margin. In Section~\ref{sec:main} we will show that this allows us to round the fractional allocation into an integral allocation that is envy-free.

We begin by introducing variations of two well-known results. Recall that the KL divergence between two $m$-dimensional probability vectors $b_i$ and $b_k$ is

\[D_{KL}(b_i \| b_k) = \sum_j b_{ij} \log(b_{ij}) -b_{ij} \log(b_{kj}).\]

In our case, we are concerned with the KL divergence between the bid vectors of agents. This value may be bounded via the well-known \textit{Pinsker's Inequality}, which can be stated as follows for discrete probability vectors.

\begin{lemma}[Pinsker's Inequality \cite{kullback1967lower,csiszar1967information}]\label{lemma:pinsker}

Let $p, q$ be two $m$-dimensional probability vectors. Then

\[\delta_{TV}(p,q) \leq \sqrt{\frac{1}{2}D_{KL}(p\| q)}\]

where $ \delta_{TV}(p,q) $ is the total variation distance, and can be calculated as

\[ \delta_{TV}(p,q) = \sum_{j \in [m]} \frac{1}{2} \bigg | p(j) - q(j)\bigg |.\]
    
\end{lemma}

In addition, we will also use the Chernoff bounds on sums of independent $[0,1]$-valued variables. We slightly modify their statement to include a scale factor $z$, as follows.

\begin{lemma}[Chernoff Bounds]\label{lemma:chernoff}
Let $X_1, \dots, X_k$ be independent random variables valued on $[0,z]$, and let $S:=\sum_{i\in[k]} X_i$. Then, for any $\epsilon \in (0,1]$,
\begin{align*}
&Pr[S \leq (1-\epsilon)E[S]] \leq exp\bigg(\frac{-\epsilon^2}{2z} E[S]\bigg), \text{and}\\
&Pr[S \geq (1+\epsilon)E[S]] \leq exp\bigg(\frac{-\epsilon^2}{3z} E[S]\bigg).
\end{align*}
\end{lemma}

Using these lemmas, we can outline our proof. Recall that each agent receives the fractional allocation $x_{ij} = \frac{\log(b_{ij}) + c}{nC} + \frac{nC - \sum_k (\log( b_{ij}) + c)}{n^2C}$. The second term is common for all agents and can be ignored when computing envy. Thus, the fractional envy margin is $fEM_{ik} = \frac{\sum_j  \bar v_{ij} \log(b_{ij}) - \bar v_{ij} \log(b_{kj})}{nC}$. We observe that the numerator is approximately the KL divergence between bids $D_{KL}(b_i \| b_k) = \sum_j b_{ij} \log(b_{ij}) -b_{ij} \log(b_{kj})$.

This gives us the core idea for our proof. First, we will show that bids are approximately equal to normalized valuations, which will let us use Pinsker's Inequality to show that the KL divergence between bids is at least $\delta^2/4$ with high probability. Then, we will show that $fEM_{ik}$ differs from the scaled KL divergence by at most a function of the lower bound $b_{min}$ that can be made relatively small compared to $\delta$. Together, this will show that when $b_{min}$ is small, with high probability $fEM_{ik}$ is at least $\frac{\delta^2}{4nC}$ for all pairs of agents $i,k$.

\subsection{Typicality}

To formalize this argument, we require a constant $\epsilon$ sufficiently smaller than $\delta$. It suffices to choose $\epsilon = \frac{1}{25}\delta$. We restrict our focus to valuations that do not deviate too far from their expected behavior. We define \textit{typicality} as follows.

\begin{definition} Let $(v_i)_{i \in N}$ be a profile of valuations. We say that $(v_i)_{i \in N}$ is \textit{typical} if, for positive constant $\epsilon < \frac{\delta}{25}$,
\begin{enumerate}[leftmargin=8mm]
    \item[T1.] $\forall i$, $\sum_j v_{ij} \in [(1-\epsilon)m\mu_i, (1+\epsilon)m\mu_i]$, and 
    \item[T2.] $\forall i \neq k$, $\sum_j |\frac{v_{ij}}{\mu_i} - \frac{v_{kj}}{\mu_k}| \geq  (1-\epsilon)\delta m$.
\end{enumerate}
\end{definition}

In our setting, asymptotic instances are typical with high probability as $m$ grows large.

\begin{restatable}{lemma}{lemtypical}
When $m = \Omega(n\log(n))$, $(v_i)_{i \in N}$ is typical with high probability.
\end{restatable}

\begin{proof}
    Both parts of this proof involve a relatively direct application of the Chernoff bounds. For the first part, note that $v_{ij} \in [0,1]$, so

    \begin{align*}
    Pr[\sum_j v_{ij} \not \in [(1-\epsilon)m\mu_i, (1+\epsilon)m\mu_i] \\
    \leq 2 exp(\frac{-\epsilon^2}{3} m \mu_i) \leq 2 exp(\frac{-\epsilon^2}{3} m \mu_l).
    \end{align*}

    Noting that asymptotically $m\geq n$, taking the union bound over all agents gives us

    \begin{align*}
    Pr[T1 \textit{ doesn't hold }] \leq 2me^{-m} exp (-\frac{\epsilon^2}{3} \mu_l)
    \end{align*}

    which converges to $0$ exponentially in $m$. Similarly, $|\frac{v_{ij}}{\mu_i} - \frac{v_{kj}}{\mu_k}| \in [0,\frac{1}{\mu_l}]$ and it is independent across different items. By assumption, $E[|\frac{v_{ij}}{\mu_i} - \frac{v_{kj}}{\mu_k}| ] \geq \delta $. Thus, denoting $S := \sum_j |\frac{v_{ij}}{\mu_i} - \frac{v_{kj}}{\mu_k}| $, we have

    \begin{align*}
    Pr[S < (1-\epsilon)m\delta] \leq Pr[S < (1-\epsilon)E[S]] \\
    \leq exp(\frac{-\epsilon^2\mu_l}{2} E[S]) \leq exp(\frac{-\epsilon^2\mu_l}{2} m \delta ).
    \end{align*}

    Taking the union bound over all pairs of agents gives

    \begin{align*}
    Pr[T2\textit{ doesn't hold }] \leq 2m^2e^{-m} exp (-\frac{\epsilon^2}{2} \mu_l \delta )
    \end{align*}

    which again converges to 0 exponentially in $m$.
\end{proof}

It turns out that that correct choice of $b_{min}$ is $\Theta(\frac{1}{m})$, so we may select an appropriate constant $l$ and define $l/m = b_{min}$ (we will specify the value of $l$ later). With typical valuations, feasible scale factors exist and are bounded. Additionally, bids $b_{ij} = s_i \bar v_{ij}$ are bounded above by $\frac{2}{m\mu_l}$.

\begin{restatable}{lemma}{lemtypicalbids} \label{lem:typicalbids}
    Let $(v_i)_{i \in N}$ be typical, and $l/m = b_{min}$. Then
    \begin{enumerate}[leftmargin=8mm]
        \item $\forall i\in N,j\in M$, $b_{ij} < \frac{2}{m\mu_l}$, and
        \item $\forall i\in N$, $s_i \in [1-l, 1]$.
    \end{enumerate}
\end{restatable}

\begin{proof}
    We first argue that in a typical instance, $max(h_i) \geq 1$ for all $i$, so every agent has a scale factor that sets its bids that sum to 1. Note that

    \[{(1-\epsilon)m\mu_i} > \frac{3}{4} m \mu_l.\]

    Since the maximum value of an item is 1, we have more than $\frac{3}{4} m \mu_l$ items that have nonzero value. If we place the maximum bid on each item, then our total sum of bids is more than

    \[\frac{3}{4} m \mu_l \frac{2}{\mu_lm} = \frac{3}{2}.\]

    Thus, $max(h_i) > 3/2$, so every agent has a scale factor $s_i$ setting bids that sum to 1. Now, note that
    
    \[\bar v_{ij} = \frac{v_{ij}}{\sum_j v_{ij}} \leq \frac{v_{ij}}{(1-\epsilon)m\mu_i}< \frac{2}{m \mu_l}.\]

    This means that no $\bar v_{ij}$ attains the upper bound $\frac{2}{m \mu_l}$. Suppose for the sake of contradiction that $s_i > 1$. Then $s_i \bar v_{ij} > \bar v_{ij}$, and we have two cases. If $\bar v_{ij} < \frac{l}{m}$, then $b_{ij} \geq \frac{l}{m} > \bar v_{ij}$. On the other hand, if $\bar v_{ij} \geq \frac{l}{m}$, then $\bar v_{ij} \in [\frac{l}{m}, \frac{2}{\mu_lm})$, and $b_{ij} = min\{s_i\bar v_{ij}, \frac{2}{\mu_lm}\}$, both of which are strictly larger than $\bar v_{ij}$. Thus, in all cases, $\bar v_{ij} < b_{ij}$, but $\sum_j \bar v_{ij} = \sum_j b_{ij} = 1$, which is a contradiction, so $s_i \leq 1$. This implies $b_{ij} < \frac{2}{\mu_lm}$.

    Now, for the lower bound on $s_i$, we don't need typicality. Note that $b_{ij} \leq max(l/m , s_i \bar v_{ij}) \leq l/m +  s_i \bar v_{ij}$. Because $\sum_j b_{ij} = \sum_j \bar v_{ij}  =1 $, we may sum over all $j$ to say
\[1 = \sum_j b_{ij} \leq l + s_i \sum_j \bar v_{ij} = l + s_i \]
\[\implies \quad s_i \geq 1 - l.\]
\end{proof}

\section{The PRD Mechanism} \label{sec:main}
We are now ready to formally specify our mechanism. We begin with the choice of $l$, which we want to be sufficiently smaller than $\delta$. It suffices to choose $l = \frac{1}{25}\delta$.
Similarly, observing Lemma~\ref{lem:typicalbids}, we may set $b_{max} = \frac{2}{m\mu_l}$, which lets us decrease the dummy bid by a factor of $m$ without significantly affecting player bids.\footnote{This choice makes the reasonable assumption that $\delta, \mu_l$ are known. If they are unknown, it suffices to substitute $l, \mu_l$ with any functions that are $o(1)$ in $m$, i.e. sufficiently small as $m$ grows. C becomes $\omega(1)$ instead of constant, which changes all our big-$\Omega$ bounds into little-$\omega$ bounds instead.} In our allocation function, this determines the values of $c$ and $C$ as $c =$ $-\log (b_{min}) =$ $-\log(l/m)$ and $C =$ $f(b_{max}) =$ $\log(\frac{2}{\mu_lm}) -$ $\log(\frac{l}{m}) =$ $\log(\frac{2}{\mu_l\cdot l})$.

The PRD Mechanism is formally described in Algorithm~\ref{alg:prd}. In the first phase, the mechanism constructs optimal bids using \textsc{Bid}-\textsc{Construction} (Algorithm~\ref{alg:bid_construction}), and then calls the \textsc{Fractional}-\textsc{Allocation} subroutine (Algorithm~\ref{alg:fractional}). In the second phase, the mechanism uses \textsc{Randomized}-\textsc{Rounding} (Algorithm~\ref{alg:rounding}) to round the fractional allocation in a randomized manner, ensuring that for typical instances the resulting allocation in envy-free.

\begin{algorithm}[H]
\caption{\textsc{The PRD Mechanism}}\label{alg:prd}
\begin{algorithmic}[1]
\State\textbf{Input} agents $[n]$, items $[m]$, reported valuations $(v_i)_{i\in[n]}$, threshold $l$.
\State\textbf{Output} allocation $(A_1, \ldots, A_n)$ that is envy-free \textit{whp}.
\Statex
\rule{.15\textwidth}{.5pt}
{First Phase}
\rule{.15\textwidth}{.5pt}
\For{$i \in [n]$}
    \State $b_i \gets \textsc{Bid-Construction}(v_i,l)$
\EndFor

\State $x \gets \textsc{Fractional-Allocation}((b_i)_{i\in[n]},l)$
\Statex 
\rule{.139\textwidth}{.5pt}
{Second Phase}
\rule{.139\textwidth}{.5pt}
\State $A \gets \textsc{Randomized-Rounding}(x)$
\Statex
\rule{.438\textwidth}{.5pt}
\State \textbf{return} $(A_1, \ldots, A_n)$
\end{algorithmic}
\end{algorithm}

\begin{algorithm}[H]
\caption{\textsc{Bid-Construction}}\label{alg:bid_construction}
\begin{algorithmic}[1]
\vspace{0.2cm}
\State\textbf{Input} reported valuations $v_i$, threshold $l$.
\State\textbf{Output} bids $b_{ij}$ for agent $i$.
\Statex
\rule{.438\textwidth}{.5pt}

\State Set $h_i'(s) \equiv 0$ 
\State Set $\bar v_i = v_i /\sum_jv_{ij}$ 
\For{$j \gets 1 $ to $m$}
    \If{$\bar v_{ij} \neq 0$}
        \State Increase $h_i'(s)$ over $s \in [\frac{l}{m}\frac{1}{\bar v_{ij}}, \frac{2}{\mu_lm}\frac{1}{\bar  v_{ij}}]$ by $\bar v_{ij}$
    \EndIf
\EndFor

\State Construct $h_i(s)$ using $h_i(0) = l$ and $h_i'$
\If{$max(h_i) \geq  1$}
    \State Find $s_i$ s.t. $h_i(s_i) = 1$ 
    \State Return $b_i = max(min( s_i \bar v_i, \frac{2}{\mu_lm}), \frac{l}{m})$
\Else
    \For{$j \gets 1 $ to $m$}
        \State if $v_{ij} > 0$, set $b_{ij} = \frac{2}{\mu_lm}$
    \EndFor
    \State Set remaining $b_{ij}$s arbitrarily such that $\sum_j b_{ij}=1$
\EndIf

\end{algorithmic}
\end{algorithm}

\begin{algorithm}[H]
\caption{\textsc{Fractional-Allocation}}\label{alg:fractional}
\begin{algorithmic}[1]
\vspace{0.2cm}
\State\textbf{Input} profile $(b_i)_{i\in[n]}$ of agent bids
\State\textbf{Output} fractional allocation $(x_1, \ldots, x_n)$
\Statex
\rule{.438\textwidth}{.5pt}
\State $C \gets -\log(l) + \log(\frac{2}{\mu_l})$
\State $c \gets -\log(l/m)$

\For{every $j\in [m]$}
    \For{$i \gets 1 $ to $n$}
        \State $x_{ij} \gets \frac{1}{nC}\bigg(\log(b_{ij}) + c\bigg)$
    \EndFor
    \State $dummyAlloc \gets \frac{1}{C} \log(\frac{2}{\mu_lm}) - \sum_i A_{ij}$
    \State $x_{\cdot j} \gets x_{\cdot j} + \Vec{\mathbf{1}} \cdot dummyAlloc/n$
\EndFor

\State return $x$
\end{algorithmic}
\end{algorithm}

\begin{algorithm}[H]
\caption{\textsc{Randomized-Rounding}}\label{alg:rounding}
\begin{algorithmic}[1]
\vspace{0.2cm}
\State\textbf{Input} Profile $(x_i)_{i\in[n]}$ of fractional allocations
\State\textbf{Output} Integral allocations $(A_1, \ldots, A_n)$ that are envy-free with high probability
\Statex
\rule{.438\textwidth}{.5pt}
\State Initialize $A = 0_{n\times m}$

\For{every $j\in [m]$}
    
    \State $x_j \gets [x_{1j} \dots x_{nj}]$
    \State $i^* \gets sample([n], x_j)$
    \State $A_{i^*j} = 1$
\EndFor

\State return $A_1, \dots A_n$
\end{algorithmic}
\end{algorithm}

\begin{remark}
    The PRD Mechanism can be implemented in polynomial time. In Algorithm~\ref{alg:bid_construction}, the function $h$ can be constructed using an appropriate data structure in polynomial time, and it is piecewise linear with $O(m)$ segments, so we may find $x$ with $h(x) = 1$ in polynomial time. Algorithms~\ref{alg:fractional} and~\ref{alg:rounding} contain simple operations whose running times are easily verified to be polynomial.
\end{remark}

We first show that our mechanism is truthful in expectation by showing that the bids found by Algorithm~\ref{alg:bid_construction} are agent-optimal. 

\begin{restatable}{theorem}{thmtruthful}\label{thm:truthful}
    The PRD Mechanism is truthful in expectation.
\end{restatable}

\begin{proof}
    First, it is easy to verify that the \textsc{Randomized}-\textsc{Rounding} step preserves the expected value of any fractional allocation, so it suffices to show that our fractional allocation is optimal among all obtainable allocations for each agent.

    Let $b_{ij}$ be arbitrary feasible bids for agent $i$. The fractional allocation $x_{ij}$ as a function of the bids is

    \[x_{ij} = \frac{\log( b_{ij}) + c}{nC} + \frac{nC - \sum_k \big (\log( b_{kj}) + c \big )}{n^2C}\]

    For the purpose of optimizing $i$'s bid, we may separate all additive terms that do not depend on $b_{ij}$ into some term $F_j(b_{-i})$, leaving us with

    \[x_{ij} = \frac{n-1}{n^2C} \log( b_{ij}) + F_j(b_{-i})\]

    Agent $i$'s normalized value can then be expressed as

    \[ \sum_j \bar v_{ij}x_{ij} = \sum_j \bar  v_{ij}\frac{n-1}{n^2C} \log( b_{ij}) +  \sum_j \bar  v_{ij} F_j(b_{-i})\]

    Agent $i$ can only influence the first term, so the optimal bid for agent $i$ is given by the solution to the following constrained optimization problem
    
    \begin{align*}
    min \ \ \ \ & \sum_j -\bar v_{ij} \log(b_{ij}) \\
    s.t. \ \ \ \ & \sum_j b_{ij}  = 1 \\ 
     & b_{ij} \geq l/m\ \  \forall j \\
     & b_{ij} \leq 2/m\mu_l\ \  \forall j 
    \end{align*}

    We first note that the objective is a convex function, and the feasible set is a convex set with linear constraints. We also note that $l/m < \frac{1}{m} < \frac{2}{m\mu_l}$, so the constraint set has non-empty interior. Thus the KKT conditions are both necessary and sufficient for optimality. We get the Lagrangian

    \[
    \sL =  \sum_j - \bar v_j \log(b_{ij}) + \lambda(1-\sum_j b_{ij}) 
    + \sum_j \mu_j (\frac{l}{m}- b_{ij}) + \sum_j \nu_j (b_{ij} - \frac{2}{m\mu_l})\]\[
    \sL_j = -\frac{\bar  v_{ij}}{b_{ij}} - \lambda - \mu_j + \nu_j= 0
    \]
    Then, we define $b_{ij}(s) = min \{max\{s \bar  v_{ij}, \frac{l}{m} \} , \frac{2}{m\mu_l}\}$, representing the projection of normalized values scaled by $s$ onto our bid interval. We show that by setting $\lambda = -(s^{-1})$, any such bid will satisfy all the KKT conditions except possibly the primal constraint $\sum_j b_{ij}(s) = 1$.
    
    First, it is clear that the upper and lower bound constraints are satisfied by our construction of $b_{ij}(s)$. For some item $j$, if neither the upper nor lower constraint is binding, then $\frac{\bar v_{ij}}{b_{ij}(s)} = s^{-1}$. Setting $\lambda = -s^{-1}$ and $\mu_j = \nu_j = 0$ satisfies $\sL_j = 0$, complementary slackness, and dual feasibility. If $b_{ij}(s) = l/m$, then $b_{ij}(s) \geq s \bar  v_{ij}$, so $-\frac{\bar  v_{ij}}{b_{ij}(s) } \geq -s^{-1} = \lambda$. Then, we observe that setting $\nu_j = 0$ and $\mu_j = -\frac{\bar v_{ij}}{b_{ij}(s) } - \lambda \geq 0$ satisfies complementary slackness, $\sL_j = 0$, and dual feasibility as desired. Similarly, if $b_{ij}(s) = \frac{2}{m\mu_l}$, then  $b_{ij}(s) \leq s \bar v_{ij}$, so $-\frac{\bar v_{ij}}{b_{ij}(s)} \leq -s^{-1} = \lambda$. Thus, setting $\mu_j = 0$ and $\nu_j = \frac{\bar v_{ij}}{b_{ij}(s)} + \lambda \geq 0$
    again satisfies complementary slackness, $\sL_j = 0$, and dual feasibility, which means that all KKT conditions except $\sum_j b_{ij}(s) = 1$ are satisfied for \textit{any} $s$.

    We may view $\sum_j b_{ij}(s)$ as a function of $s$, denoted $h_i(s)$. We note that $h_i(s)$ is a piecewise linear function, where each segment has slope equal to the total normalized value of items whose bids are in the interior of the feasible set. Specifically, each nonzero $\bar v_{ij}$ increases the slope of $h_i(s)$ over $[\frac{l}{m} \frac{1}{\bar v_{ij}}, \frac{2}{\mu_lm} \frac{1}{\bar v_{ij}}]$ by $\bar v_{ij}$. If we set $s = 0$, then $\forall j, s\bar v_{ij} < l/m$, so $h_i(s) = l < 1$. From here, we have two cases. On one hand, if $max(h_i) \geq 1$, then because $h_i(s)$ is a continuous non-decreasing function of $s$, there exists $s_i$ s.t. $h_i(s_i) =  1$, which is a KKT point, and thus optimal. We denote by $s_i$ the scale factor for agent $i$. On the other hand, if $max(h_i) < 1$, then this means even for arbitrarily large $s$, we cannot meet the budget constraint. This only occurs if there are so few non-0 valuations that agent $i$ may bid the maximum on all such items and still have some budget left over. In this case, it is easy to see that assigning the maximum bid to all positively valued goods, and assigning the remaining budget arbitrarily, is optimal. 

    Our \textsc{Bid}-\textsc{Construction} algorithm uses this to compute optimal bids. The function $h_i(s)$ is piecewise-linear with at most $2m+1$ pieces, so we can construct it in polynomial time. The maximum value of $h_i(s)$ can be computed by setting $\bar v_{min} = min\{\bar v_{ij}: \bar v_{ij} > 0\}$ and $s = \frac{2}{\mu_lm} \frac{1}{\bar v_{min}}$. Then, if $max(h_i) \geq 1$, we find $s_i$ s.t. $h_i(s_i) = 1$ and set bids equal to $b_{ij}(s_i)$. Otherwise, we assign a maximum bid to all positively-valued goods and assign the remaining budget arbitrarily, maintaining truthfulness in the remaining case.
\end{proof}

Next, we show that typical instances yield fractional allocations with a high fractional envy margin. We note that conditioning on valuations being typical results in all values being correlated. We therefore emphasize that after assuming typicality, the theorem and proof are both deterministic.

\begin{restatable}{lemma}{lemfem}\label{lem:fem}
    Let $(v_i)_{i \in N}$ be typical. Then, $\forall i \neq k$,
    \[fEM_{ik} \geq \frac{\delta^2}{4nC}.\]
\end{restatable}

\begin{proof}

Recall that the fractional envy margin is $fEM_{ik} = \frac{\sum_j  \bar v_{ij} \log(b_{ij}) - \bar v_{ij} \log(b_{kj})}{nC}$. To compute a lower bound on this value, we will need to compute a bound on the KL divergence between bids $D_{KL}(b_i\|b_k)$.

We first want to bound the total variation distance between $b_i$ and $b_k$. As our guarantee is on the difference between normalized valuations, we first define estimated normalized values $v_{ij, est} := \frac{v_{ij}}{m \mu_i}$. By the triangle inequality, we have

\begin{align*}
\sum_j |\bar v_{ij} - \bar v_{kj}| \geq &\sum_j |\bar v_{ij,est} - \bar v_{kj,est}| \\
&- \sum_j |\bar v_{ij, est} - \bar v_{ij} | \\
&- \sum_j |\bar v_{kj, est} - \bar v_{kj} |.
\end{align*}

By the second typicality condition, the first term is at least $(1-\epsilon)\delta$. Moreover, 

\begin{align*}
\bar v_{ij} &= \frac{v_{ij}}{\|v_i\|_1} \in [\frac{v_{ij}}{m\mu_i(1+\epsilon)}, \frac{v_{ij}}{m\mu_i(1-\epsilon)}] \\
&= [\frac{\bar v_{ij, est}}{(1+\epsilon)}, \frac{\bar v_{ij, est}}{(1-\epsilon)}]
\end{align*}

meaning that $ |\bar v_{ij, est} - \bar v_{ij} | \leq \frac{\epsilon}{1-\epsilon} \bar v_{ij, est} \leq \frac{\epsilon}{(1-\epsilon)^2} \bar v_{ij}$. A similar argument applies for the third term. As $\sum_j \bar v_{ij} = 1$, the second and third terms are each bounded above by $\frac{\epsilon}{(1-\epsilon)^2}$, giving us

\[  \sum_j |\bar v_{ij} - \bar v_{kj}| \geq \delta - \epsilon\delta\ -2 \frac{\epsilon}{(1-\epsilon)^2}.\]

Now, again by the triangle inequality, we have

\begin{align*}
\sum_j |b_{ij} - b_{kj}| &\geq \sum_j |\bar v_{ij} - \bar v_{kj}|\\
&- \sum_j |b_{ij} - \bar v_{ij}| \\
&- \sum_j |b_{kj} - \bar v_{kj}|.
\end{align*}

We split the analysis into two cases. If $b_{ij}$ is equal to its minimum value, i.e. $b_{ij} = l/m$, then $s_i\bar v_{ij} \leq l/m$. Because $s_i \geq 1-l$, $\bar v_{ij} \in [0,\frac{l}{(1-l)m}] \subset [0, \frac{2l}{m}]$. Thus, $b_{ij}$ differs from $\bar v_{ij}$ by no more than $l/m$. On the other hand, if $b_{ij} > l/m$, then $s_i\bar v_{ij} = b_{ij} \in [(1-l)\bar v_{ij} , \bar v_{ij}]$, so $b_{ij}$ differs from $\bar v_{ij}$ by at most $ l \bar v_{ij}$. In either case, for all $j$, $|b_{ij} - \bar v_{ij}| \leq l/m + l\bar v_{ij}$. The same result holds for $k$ instead of $i$. As $\sum_j l \bar v_{ij} + l/m = 2l$, we have

\[ \sum_j |b_{ij} - b_{kj}| \geq \delta   - \epsilon\delta\ -2 \frac{\epsilon}{(1-\epsilon)^2}- 4l.\]

For our choice of $l$ and $\epsilon$, i.e. when both $l$ and $\epsilon$ are less than $\delta/25$, for any $\delta\in(0,1)$ the above inequality gives us

\[ \sum_j |b_{ij} - b_{kj}| \geq \frac{1}{\sqrt{2}}\delta.\]

By Pinsker's inequality we have

\begin{align*}
D_{KL}(b_i\|b_k) &\geq 2 \delta_{TV}^2 (b_i, b_k) \\
&= \frac{1}{2} \bigg ( \sum_j |b_{ij} - b_{kj}| \bigg )^2 \\
&\geq \frac{1}{4}\delta^2.
\end{align*}

Thus we also have

\[ \frac{1}{nC}D_{KL}(b_i\|b_k) \geq \frac{1}{4nC} \delta^2.\]

How do we relate the above quantity to the fractional envy margin? We again split the analysis into two cases. If $b_{ij}$ is not equal to its minimum value, then $s_i \bar v_{ij} = b_{ij}$ and

\[\bar v_{ij} \log(b_{ij}) - \bar v_{ij} \log(b_{kj}) =  \frac{b_{ij}}{s_i}\log(b_{ij}) - \frac{b_{ij}}{s_i} \log(b_{kj}).\]

On the other hand, if $b_{ij}$ is minimal, then $\log(b_{ij})$ is minimal as well. This means $\log(b_{ij}) - \log(b_{kj}) \leq 0$, and $s_i \bar v_{ij} \leq l/m$, so

\[s_i\bar v_{ij} \leq b_{ij} \implies \bar v_{ij} \leq \frac{b_{ij}}{s_i}\]
\[\implies \bar v_{ij} (\log(b_{ij}) - \log(b_{kj})) \geq \frac{b_{ij}}{s_i}  (\log(b_{ij}) - \log(b_{kj})).\]

This inequality holds in both cases. Thus, noting that $1/s_i > 1$ and summing over all $j$, we have

\begin{align*}
fEM_{ik} &= \frac{\sum_j  \bar v_{ij} \log(b_{ij}) - \bar v_{ij} \log(b_{kj})}{nC}\\
&\geq \frac{1}{s_i}\frac{\sum_j b_{ij} \log(b_{ij}) - b_{ij} \log(b_{kj})}{nC}\\
&\geq \frac{\sum_j b_{ij} \log(b_{ij}) - b_{ij} \log(b_{kj})}{nC}\\
&= \frac{1}{nC}D_{KL}(b_i\|b_k)\\
&\geq \frac{1}{4nC} \delta^2.
\end{align*}
\end{proof}

Finally, we prove that typicality and the resulting $fEM$ bound give us envy-freeness after rounding.

\begin{restatable}{theorem}{thmef}\label{thm:ef}
    Let $(v_i)_{i \in N}$ be typical. When $m = \Omega(n\log(n))$, allocating each item $j$ independently with probabilities $x_{ij}$ yields an envy-free allocation with high probability.
\end{restatable}

\begin{proof}
For a given typical valuation profile and associated fractional allocation $x_{ij}$, suppose the items are rounded independently with probabilities given by $x_{ij}$. Let $X_{ij}((v_i)_{i \in N})$ be the $\{0,1\}$-valued indicator variable for whether agent $i$ is assigned item $j$. We note that the $X_{ij}$'s are independent across items, as they are defined with respect to a given valuation profile.

For our analysis of the rounding step, we select some $\delta' \leq \min\{\delta, 2\sqrt{C}\}$. Note that the previous bound on $fEM_{ik}$ continues to hold if we replace $\delta$ with $\delta'$. Thus, by Lemma~\ref{lem:fem},

\[\forall i \neq k, \bar v_ix_i \geq \bar v_ix_k +  \frac{\delta'^2}{4nC}. \]

Note that $\sum_{k\in [n]} \bar v_ix_k   = \bar v_i \cdot \Vec{1} = 1 $. Summing the previous inequality over all agents and substituting gives us

\[ n \bar v_ix_i \geq 1 + \frac{n-1}{4nC}\delta'^2 \implies \bar v_i x_i \geq \frac{1}{n} + \frac{n-1}{4n^2C}\delta'^2 \]

For $n \geq 2$, this gives the inequality

\[ \bar v_i x_i \geq \frac{1}{n} + \frac{\delta'^2 }{8nC}\]

which lower bounds the expected value of $\sum_j v_{ij}X_{ij}$. Additionally, note that no agent can receive more than $2/n$ on any item and so $\bar v_i x_i \leq 2/n$. This is because, in the extreme case, if agent~$i$ `bids' the maximum value and every other agent `bids' the minimum, agent~$i$ receives 
\[\frac{C}{nC} + \frac{nC-C}{n^2C} \leq \frac{2}{n}.\]

We want to obtain similar bounds on $\bar v_i x_k$. To do this, we will instead consider an auxiliary allocation that is strictly better than agent $k$'s actual allocation and that suffices to bound agent $i$'s envy. We define $\Tilde{x}_k$ such that $\forall j, x_{kj} \leq \Tilde{x}_{kj} \leq \frac{2}{n} $ and $\bar v_ix_i - \bar v_i\Tilde{x}_k$ attains the fractional envy margin bound exactly, i.e.,

\[\bar v_ix_i = \bar v_i \Tilde{x}_k +  \frac{\delta'^2}{4nC}.\]

Thus we have $\bar v_i \Tilde{x}_k \geq \frac{1}{n} - \frac{\delta'^2}{8nC}$. Note that by our choice of $\delta'$, we also have $\delta'^2/8C \leq \frac{1}{2}$. Thus $\bar v_i \Tilde{x}_k \geq \frac{1}{2n}$. Additionally, as before we have $ \bar v_i \Tilde{x}_k  \leq \frac{2}{n}$. Then, using $\Tilde{x}_k$, we define the $\{0,1\}$-valued indicator variables $\Tilde{X}_{kj}$ that equal $1$ with probability $\Tilde{x}_{kj}$, are independent across items, and stochastically dominate $X_{kj}$. This gives us the bound

\[ EM_{ik} \geq \sum_j \bar v_{ij} X_{ij} - \sum_j  \bar v_{ij} \Tilde{X}_{kj}\]

Now, the above envy margin is ex-post nonnegative when both of the following inequalities hold.
\begin{align*}
    &\sum_j\bar v_{ij} X_{ij} \geq \bar v_ix_i - \frac{\delta'^2}{8nC}, \text{ 
    and} \\
    &\sum_j\bar v_{ij} \Tilde X_{ij} \leq \bar v_i \Tilde x_k + \frac{\delta'^2}{8nC}.
\end{align*}

We first bound the probability that the first inequality is false. Choose $\epsilon' = \min\{\frac{1}{2}, \frac{\delta'^2}{16C}\} $ so that 

\[\frac{2\epsilon'}{n} \leq \frac{\delta'^2}{8nC}\]

Note that $\bar v_i x_i \leq \frac{2}{n}$, and $\epsilon' > 0$, so 

\[\bar v_i x_i \epsilon' \leq \frac{2\epsilon'}{n} \leq \frac{\delta'^2}{8nC}\]

which implies

\begin{align*}
& Pr[\bar v_i X_i \leq  \bar v_ix_i - \frac{\delta'^2}{8nC} ]\\
\leq\; & Pr [\bar v_i X_i \leq   (1-\epsilon') \bar v_ix_i] 
\end{align*}

Finally, by typicality, we know that $\bar v_{ij} \leq \frac{2}{\mu_l m}$, and $\bar v_ix_i \geq \frac{1}{n}$. Then, noting that $E[\bar v_i X_i] = \bar v_ix_i$, we apply the first Chernoff bound with $z=\frac{2}{\mu_l m}$ and $\bar c_1 = 12\frac{1}{\epsilon'^2\mu_l}$, and for $m \geq \bar c_1 \cdot n \log{n}$ we get

\begin{align*}
    Pr[\bar v_i X_i \leq (1-\epsilon') \bar v_ix_i] &\leq exp(-\frac{\epsilon'^2}{2} \frac{\mu_lm}{2} \frac{1}{n}) \\
    &\leq exp(-3\log{n}).
\end{align*}

By similar reasoning, noting that $\bar v_i\Tilde{x}_k \geq \frac{1}{2n}$ applying the second Chernoff bound with $z=\frac{2}{\mu_l m}$ and $\bar c_2 = \frac{36}{\epsilon'^2\mu_l}$, for $m \geq \bar c_2 \cdot n \log{n}$ we also have

\begin{align*}
Pr[\bar v_{i} \Tilde{X}_{k}  \geq  \bar v_i\Tilde{x}_k  +  \frac{\delta'^2 }{8nC}] &\leq\; Pr[ \bar v_{i} \Tilde{X}_{k} \geq (1+\epsilon')  \bar v_i\Tilde{x}_k] \\
&\leq\; exp(-\frac{\epsilon'^2}{3} \frac{1}{2n} \frac{m \mu_l}{2})\\
&\leq exp(-3\log{n}).
\end{align*}

Now, letting $\bar c = \max\{\bar c_1,\bar c_2\} = \frac{36}{\epsilon'^2\mu_l}$, the probability that agent $i$ envies agent $k$ is at most $2exp(-3\log{n})$. Taking the union bound over all pairs of agents gives us that the probability any agent is envious is at most $2n^2 exp(-3 \log(n)) = \frac{2}{n}$ which goes to 0 as $n$ goes to infinity, so as long as $m \geq \bar c n \log{n}$. Note that if $n$ is bounded, then the probability is a negative exponential in $m$. Thus, in either case, the resulting allocation is envy-free with high probability.
\end{proof}

\section{Extensions to Weights, Groups, and Types} \label{sec:extensions}

\subsection{Weighted Envy-Freeness} \label{subsec:wef}

In the weighted fair division problem, each agent has an associated weight $w_i$. An allocation is Weighted-Envy-Free (WEF) if $\forall i,k \ \frac{v_i(A_i)}{w_i} \geq \frac{v_i(A_k)}{w_k}$. In recent work, \citet{manurangsi2025weighted} showed that if the ratios between weights remain bounded by a constant as $n$ grows to infinity, then a WEF allocation exists with high probability when $m=\Omega(n\log(n)/\log(\log(n)))$. We extend this literature by allowing both to grow simultaneously. In particular, for $W = \sum_i w_i$, when $\frac{W}{w_i} \leq \rho$, our mechanism can be adapted to obtain envy-freeness with high probability when $m=\Omega(\rho \log(n))$, while maintaining truthfulness-in-expectation. One consequence is that we allow for agents with weights linear in $n$ (such as a government agency or union entitled to a constant fraction of the total utility). Moreover, if there are a constant number of such agents, $\rho$ is still $O(n)$, so WEF is achievable in this setting with the same asymptotic bound on $m$.

\begin{restatable}{theorem}{thmweights}\label{thm:weights}
    Let $\rho \geq \frac{W}{w_i}$. In our asymptotic fair division setting, when $m = \Omega(\rho \log(n))$, with high probability there exists a weighted envy-free allocation that can be obtained in polynomial time via a truthful mechanism. 
\end{restatable}

Recall that an allocation is Weighted-Envy-Free (WEF) if $\forall i,k \ \frac{v_i(A_i)}{w_i} \geq \frac{v_i(A_k)}{w_k}$. We can define the weighted envy margin as
\[EM_{ik} = \frac{v_i(A_i)}{w_i} -\frac{v_i(A_k)}{w_k}.\]

For ease of notation, we denote $W = \sum_i w_i$. We also define $\rho$ to be an upper bound on $\frac{W}{w_i}$. Clearly, $\rho \geq n$. We may then modify Algorithm~\ref{alg:prd} as follows.

\begin{enumerate}
    \item For each good $j$, agent $i$'s bid is $b_{ij}$ and it receives a share of $w_i(\log(b_{ij}) + c)$.
    \item The dummy agent receives a total share of $WC - \sum_i w_i(\log(b_{ij}) + c)$, which it splits among the agents in proportion to their weights. 
    \item All agents receive a fraction of every good proportional to their share, by independently rounding each good. 
\end{enumerate}

\begin{proof}[Proof of Theorem~\ref{thm:weights}]
It is easy to see that this mechanism is incentive compatible for the same reason as before. Agent allocations are separable functions, and scaling by a constant factor $w_i$ does not change an agent's behavior. In fact, this also shows that bids $b_{ij}$ are fully independent of weights. We can express agent $i$'s fractional allocation as follows. 

\[x_{ij} = w_i\frac{\log(b_{ij}) + c}{WC} + \frac{w_i}{W}\frac{WC - \sum_i w_i(\log(b_i) + c)}{WC}.\]

Then, we can write the fractional envy margin

\begin{equation}
    fEM_{ik} = \sum_{j\in M_l} \frac{v_{ij} \log(b_{ij}) - v_{ij} \log(b_{kj})}{WC}
\end{equation}

which is almost identical to the non-weighted case, except we divide by $WC$ instead of $nC$. Since optimal bids are independent of weights, they are identical to the non-weighted case. As such, we have with high probability for all $i\neq k$

\[fEM_{ik} \geq \frac{\delta^2}{4WC}.\]

Let $\delta' \leq \min\{\delta, 2\sqrt{C}\}$ as before. When our allocation is typical, if we allocate each item independently at random, we have

\[\forall i \neq k, \bar v_ix_i \frac{w_k}{w_i} - \frac{\delta'^2w_k}{4WC}\geq \bar v_ix_k.\]

which implies that

\begin{align*} 1 &= \sum_{k} \bar v_i x_k \leq \bar v_i x_i + \sum_{k \neq i} \bigg (\bar v_ix_i \frac{w_k}{w_i} - \frac{\delta'^2w_k}{4WC} \bigg )\\
&= \bar v_i x_i \frac{W}{w_i} - \delta'^2\frac{W - w_i}{4WC}\\
&\implies \frac{\bar v_i x_i}{w_i} \geq \delta'^2 \frac{W - w_i}{4W^2C}+ \frac{1}{W}.
\end{align*}

This is value is at least $\frac{1}{W}$, and as before, no agent may receive more than $\frac{2w_i}{W}$ of any item, so $ \frac{\bar v_i x_i}{w_i}  \leq \frac{2}{W}$. As before, we may define $\Tilde{x_k}$ so that the envy margin bound is exactly attained, and observe that $\delta'^2/4C \leq \frac{1}{2}$, so that 

\[\frac{\bar v_i \Tilde{x_k}}{w_k} \in \Big [ \frac{1}{2W}, \frac{2}{W}\Big ].\]

We choose the same $\epsilon' = min\{\frac{1}{2}, \frac{\delta'^2}{16C}\} $, and note that $\bar v_{i} x_{i} / w_i \leq \frac{2}{W}$ so that 

\[ \frac{\bar v_i x_i}{w_i}  \epsilon'\leq \frac{2\epsilon'}{W} \leq \frac{\delta'^2}{8WC}.\]

By typicality, we note that $\bar v_{ij}/w_i \leq \frac{2 }{\mu_l m w_i}$, and we know $\frac{\bar v_i x_i}{w_i} \geq \frac{1}{W}$. As in the unweighted case, we may define $\bar c_1 = \frac{12}{\epsilon'^2 \mu_l}$. Then, when $m \geq \bar c_1 \rho \log{n}$, we have

\begin{align*}
Pr[\frac{\bar v_{i}}{w_i} X_{i}  \leq \frac{v_ix_i}{w_i} - \frac{\delta'^2}{8WC}] &\leq Pr[\frac{\bar v_{i}}{w_i} X_{i} \leq (1-\epsilon') \frac{v_ix_i}{w_i}] \\
&\leq exp(-\frac{\epsilon'^2}{2} \frac{1}{W} \frac{m\mu_lw_i}{2}) \\
& \leq exp(- 3 \log{n}).
\end{align*}
A similar result holds for the second half, for an appropriately chosen $\bar c_2$. Then, defining $\bar c = \max\{\bar c_1,\bar c_2\}$, and applying the union bound, we have that when $m \geq \bar c \rho ln(n)$, the probability that any agent is envious is 

\[ \leq n^2 exp(-3 \log{n}) = \frac{2}{n}\]

so the resulting allocation after rounding is weighted-envy-free with high probability.
\end{proof}

\subsection{Weighted Envy-Freeness for Groups} \label{subsec:wefg}

Our method may also be applied to the problem of envy-freeness for groups. In this setting, each of the $n$ agents is assigned to a group $g$, which has $n_g$ members and group weight $w_g$. Allocation bundles $A_g^G$ are assigned to groups instead of agents, and every item in a group's allocation is shared as a common good among its members. Accordingly, an agent $i$ in group $g$ is weighted-envy-free if

\[\forall g' \neq g, \frac{v_i(A_g^G)}{w_g} \geq \frac{v_i(A_{g'}^G)}{w_{g'}}\]

and we say an allocation is \textit{weighted envy-free for groups} (WEFG) if this condition is satisfied for all agents. In this setting, \citet{manurangsi2017asymptotic} showed that for groups with equal sizes and equal weights, when $m = \Omega(n \log n)$, WEFG allocations exist with high probability. For a stronger set of distributional assumptions than our previous results, we generalize this result by allowing for a variable number of agents within each group and unequal weights across groups.

\begin{restatable}{theorem}{thmweightsgroups} \label{thm:weightsgroups}
    Suppose $\sD$ is well-behaved and its marginal distributions are i.i.d., and let $\beta \geq n_g$ and $\rho \geq \frac{W}{w_g}$ for every group $g$. Then, if $m = \Omega(\beta^2 ln^3(\beta)\rho ln(n))$, with high probability there exists an allocation that is weighted envy-free for groups that can be obtained in polynomial time via a truthful mechanism.
\end{restatable}

\subsubsection{Assumptions and Algorithm}

To make progress in the groups setting, we require that agent valuations are i.i.d. for each item. We additionally assume as before that $\mu  > 0$ and $E[\frac{1}{\mu}|v_{ij} - v_{kj}|] \geq \delta > 0$, so we may apply some results from previous sections.

When $\mu > 0$, because $v_{ij}$ and $v_{kj}$ are i.i.d., we may bound

\begin{align*}
    E \Big [\frac{1}{\mu}|v_{ij} - v_{kj}| \Big] & \geq E[|v_{ij} - v_{kj}|] \\
    & \geq E[(v_{ij} - v_{kj})^2] \\
    & = E[v_{ij}^2] + E[v_{kj}^2] - 2E[v_{ij}v_{kj}] \\
    & = 2\Big (E[v_{ij}^2] - E[v_{ij}]^2 \Big) \\
    & = 2Var(\sD_i).
\end{align*}

Thus our assumptions hold whenever $\sD_i$ has positive mean and variance. We additionally note that

\[ E[\frac{1}{\mu}|v_{ij} - v_{kj}|] = \frac{1}{\mu}E[|(1-v_{kj})  -(1 - v_{ij})|] \leq 2\frac{1-\mu}{\mu}\]

so $\mu < 1$.

This setting requires some new notation for our algorithm. We will use the superscript $G$ to denote any quantity as it applies to the problem for groups, so $x_g^G$ and $A_{g}^G$ are group fractional and integral allocations respectively. We define $\beta$ to be an upper bound on the maximum group size, and $\rho$ to be an upper bound on $\frac{W}{w_g}$, both of which may vary with $n$. We also note that unlike in the previous section, $\rho$ and $\beta$ may each be individually constant, but $\beta \rho \geq n$. Our proof of Theorem~\ref{thm:weightsgroups} relies on the following algorithm.

\begin{enumerate}
    \item For each group $g$, we assign each agent within that group an \textit{individual weight} $\frac{w_g}{n_g}$.
    \item We choose a sufficiently small\footnote{$l = \delta/25$ no longer suffices. Instead, we set $l = \frac{\delta^2 \mu}{32 \cdot 16}$.} constant $l$, and run the previous algorithm for weighted envy-freeness using the individual weights and lower bound $l/\beta$ to get the agent fractional allocations $x_{ij}$. We then round the above allocation independently for each item, as before.
    \item For each group $g$, we pool together the items given to each agent in that group.
\end{enumerate}

In the analysis of the above algorithm, we will also analyze the fractional allocation of item $j$ to all agents in group $g$, denoted by $x_{gj}^G = \sum_{i\in g} x_{ij}$. This mechanism is truthful in expectation for similar reasons as before. Pooling together items as we do above alters the optimization problem for an agent, since increasing its bid on item $j$ will decrease the dummy's share to each of its other group members. This means that for agent $i$ in group $g$, $x_{gj}^G$ may be expressed as

\[x_{gj}^G = \frac{n-n_g}{n^2C} \log( b_{ij}) + F_j(b_{-i})\]

for some function $F_j$. We observe that this again reduces to optimizing for $\sum \bar v_{ij} \log(b_{ij})$, so the bids produced by the same bid construction algorithm (Algorithm~\ref{alg:bid_construction}) remain optimal for agents.

\subsubsection{Proof Overview}

For the rest of this section we will consider an arbitrary agent $i$ in group $N_1$ with weight $w_1$, and an opposing group $N_2$ with weight $w_2$. Suppose we have agent $a \in N_1$ and agent $k \in N_2$. In this section, we say an event holds \textit{with high probability} if it holds with probability at least $1 - \frac{1}{n^3}$, or exponentially decreasing in $m$, so that our results are preserved under the union bound. Then, WEFG holds when the following condition is satisfied for $i$, $a$, and $k$.

\[\sum_{a \in N_1} \frac{v_i(A_{a})}{w_1} \geq \sum_{k \in N_2} \frac{v_i(A_{k})}{w_2}.\]

We define the envy margin between $i$ and $N_2$ as the difference between the left and right side terms of the above inequality, and define $fEM_{i2}$ as the corresponding fractional envy margin. We want to prove that $EM_{i2} \geq 0$ with high probability.

To do this, we want to analyze each agent's contribution to $EM_{i2}$ individually. Let

\[ EMC_{aik} := \frac{n_1}{w_1} v_i(A_a) - \frac{n_2}{w_2} v_i(A_k) \]

be agent $a$'s contribution to agent $i$'s envy margin, and define $fEMC_{aik}$ analogously. Note that $\frac{w_g}{n_g}$ is each agent's individual weight, so $EMC_{aik}$ is agent $i$'s view of the relative values of $A_a$ and $A_k$ after step 2 of the algorithm. $EM_{i2}$ is then the average of the $EMC_{aik}$ values, i.e.,

\[EM_{i2} = \frac{1}{n_1n_2} \sum_{k \in N_2}\sum_{a\in N_1} EMC_{aik}.\]

For all opponents $k$, we know that $EMC_{iik}$ is simply agent $i$'s envy margin against $k$ before pooling. Thus, from our previous results, we would expect $EMC_{iik}$ to be a positive constant. On the other hand, $EMC_{aik}$ represents agent $i$'s view of the relative value of $A_a$ against $A_k$. As $\sD_a$ is identical to $\sD_k$, and both are independent from agent i's valuations, we would expect this difference to be asymptotically small as $m\rightarrow \infty$.

We formalize this intuition through a series of lemmas, whose proofs together constitute the full proof of Theorem~\ref{thm:weightsgroups}. As $EMC_{iik}$ is very similar to our weighted envy margins from section \ref{subsec:wef}, the proof of the subsequent lemma follows similar lines.

\begin{lemma} \label{lemma:EMC_iik}
    For a given typical valuation profile $(v_i)_{i\in N}$, when $m = \Omega(\beta \log^2(\beta) \rho \log(n))$, $EMC_{iik} \geq \frac{\delta'^2}{8WC}$ with high probability.
\end{lemma}

To prove that $EMC_{aik}$ is small, we will need more structure. Like before, we want to restrict our focus to a set of typical valuations. However, for this problem we will require a strictly stronger notion, which we call \textit{WEFG-Typicality}. Then, similar to our fractional envy margin bound from previous sections, we will first show that $fEMC_{aik}$ is small when valuations are WEFG-typical.

\begin{lemma} \label{lemma:fEMC_aik}
     Let $(v_i)_{i\in N}$ be WEFG-typical. Then
    
    \[\sum_{k \in N_2}\sum_{a\in N_1\setminus\{i\}} fEMC_{aik} \geq -\frac{-\delta'^2 n_2}{16WC}.\]
\end{lemma}

Then, by applying Chernoff bounds, we will show that the integral allocations are close to the fractional allocations. 

\begin{lemma} \label{lemma:EMC_aik}
    For a given WEFG-typical valuation profile $(v_i)_{i\in N}$, when $m = \Omega(\beta^2  ln^2(\beta) \rho ln(n))$, the total difference between $EMC_{aik}$ and $fEMC_{aik}$ for all $a\neq i$ and all $k$ is no more than $ \frac{\delta'^2n_2}{16WC}$ with high probability.
\end{lemma}

Combining all three of these bounds will show that $EM_{i2}$ is non-negative, implying that the resulting allocation is envy-free.

\subsubsection{Bounds on $\mathbf{EMC_{iik}}$}

Lemma~\ref{lemma:EMC_iik} can be first proved without any new definitions.
\begin{proof} We first note that unlike before, the lower bound is now $l/\beta$. This means that $C = \log(\frac{2}{\mu_l}) - \log(l) + \log(\beta)$, and so $C$ is now a function of $\beta$. This does not affect our analysis of optimal bids, as agents cannot influence $\beta$. Moreover, because our assumption set is strictly stronger than before, we know that with high probability the profile of valuations is typical, and so for $\delta' \leq min\{\delta, 2\sqrt{\log(\frac{2}{\mu_l}) - \log(l)}\}$, with high probability

\[fEMC_{iik} \geq  \frac{\delta'^2}{4WC}.\]

Additionally, we may modify the proof from Section~\ref{subsec:wef} by substituting $\frac{v_ix_i}{w_i} - \delta'^2/8WC$ with $\frac{v_ix_i}{w_i} -\delta'^2/16WC$. Then, for $\epsilon' = min\{\frac{1}{2}, \frac{\delta'^2}{32C}\}$, we observe that

\[EMC_{iik} \geq \frac{\delta'^2}{8WC}\]

with probability at least

\[1 - 2 exp(\frac{-\epsilon'^2}{3}\frac{1}{2W}\frac{m \mu_l w_1}{2n_1}).\]

As $\epsilon' = \Omega(\log^{-1}(\beta))$, and $n_1 \leq \beta$, this is true with probability at least $1-\frac{1}{n^3}$ when $m = \Omega(\beta \log^2(\beta) \rho \log(n))$, giving us the desired result.
\end{proof}

\subsubsection{WEFG-Typicality}

Before proving the next two lemmas, we will first introduce \textit{WEFG-typicality}, and show that our new conditions are satisfied with high probability. To motivate our definition for WEFG-typical, we first observe that when $n_1$ increases, the variance of $v_iA^G_1$ grows relative to its mean, so we need some method by which to counteract this. Our method is simply to replace each $\epsilon$ with $\epsilon/\beta$, which will bound the error appropriately. Additionally, we note that when $a\neq i$, $\bar v_i x_a$ is not as well-behaved as before. In particular, $v_{aj}$ may be very small while $v_{ij}$ is large, meaning $v_{ij}\log(v_{aj})$ can approach $-\infty$. To address this, we observe that when $s_a\bar v_{aj}$ is $< l/m$, increasing $v_{aj}$, up to a point, will not actually change the bids. Thus we may impose a constant floor on valuations, which will serve to bound $\log(v_{aj})$.

Let $v^f_{ij} := max \{v_{ij} , \frac{l}{\beta} \mu (1-\epsilon/\beta)\}$, representing \textit{floored values}, and let $\bar v_{ij}^f = \frac{v_{ij}^f}{\|v_i\|_1}$ (note that we are normalizing with respect to $v$, not $v^f$). Importantly, because the floor is independent of $m$, $v_{ij}^f$ is independent across different values of $j$. Let $\mu_{f} := E[\log(v_{ij}^f)]$, and note that $\mu_f$ varies with $\beta$. As $v^f_{ij} \in [\frac{l}{\beta} \mu (1-\epsilon/\beta), 1]$, $\log(v^f_{ij})$ is bounded between $[\log(\frac{l}{\beta}\mu (1-\epsilon/\beta)), 0]$. Denoting $\bar c_f = -ln(\frac{\mu l}{2})$, we have that $|ln(v_{ij}^f)| \leq \bar c_f + \log(\beta)$. We may now formally define WEFG-typicality.

\begin{definition}
    A valuation profile $(v_i)_{i \in N}$ is WEFG-typical if, for some sufficiently small constant $\epsilon$,\footnote{As with $l$, $\epsilon = \delta/25$ is no longer sufficient. Instead, we will require $\epsilon \leq \frac{\delta^2 \mu}{32 \cdot 48}$ and $\epsilon \leq \frac{\delta^2 \mu}{32 \cdot 8} \bar c_f$, where $\bar c_f = -\log(\frac{\mu l}{2})$ and $l = \frac{\delta^2 \mu}{32 \cdot 16}$.} 

    \begin{enumerate}[leftmargin = 8mm]
        \item[W1.] $\forall i$, $\sum_j v_{ij} \in [(1-\epsilon/\beta)m\mu, (1+\epsilon/\beta)m\mu]$,
        \item[W2.] $\forall i \neq k$, $\sum_j \frac{1}{\mu}|v_{ij} - v_{kj}| \geq  (1-\epsilon/\beta)\delta m$, and
        \item[W3.] $\forall i$, $\sum_j \log(v^f_{ij}) \in [(1+\frac{\epsilon}{\beta \max\{\log(\beta), 1\}}) m \mu_f, \\ (1-\frac{\epsilon}{\beta \max\{\log(\beta), 1\}}) m \mu_f]$.
    \end{enumerate}
\end{definition}

The first two conditions are strengthed versions of the typicality assumptions and show that WEFG-typical valuation sets are also typical. The third condition is new and relates to the floored values. Importantly, when the first condition holds, $v_{ij} \leq \frac{l}{\beta}\mu (1-\epsilon/\beta)$ only when $\bar v_{ij} \leq l/(m\beta)$, and $s_i \leq 1$ remains true for WEFG-typical valuation sets. Thus, the floor only affects values for which agents would desire to place minimal bids, and does not increase those values beyond the threshold for minimal bids. Because typical valuation sets also do not attain the upper bound, this means that for WEFG-typical valuations, $b_{ij} = max\{s_i\bar v_{ij}, l/(m\beta)\} = max\{s_i\bar v_{ij}^f, l/(m\beta)\}$.

\begin{lemma}
    When $m = \Omega(\beta^2 ln^3(\beta)ln(n))$, $(v_i)_{i \in N}$ is WEFG-typical \textit{whp}.
\end{lemma}

\begin{proof}
By similar applications of the Chernoff bounds as before, we can see that the first two conditions are satisfied with high probability when $m = \Omega(\beta^2 \log(n))$. 

Now, for any $\epsilon_1 > 0$, the probability that $v_{ij} > 1-\epsilon_1$ is no greater than $\frac{\mu}{1-\epsilon_1}$. Choosing $\epsilon_1$ small enough that $1-\epsilon_1$ is above the floor, it is easy to see that this bound holds for $v_{ij}^f$, which means that

\[\mu_f \leq \frac{1-\mu-\epsilon_1}{1-\epsilon_1} ln(1-\epsilon_1).\]

As $\mu < 1$, we may choose $\epsilon_1 \leq \frac{1}{2}(1 - \mu)$, which gives us a lower bound on $|u_f|$ which is $>0$ and independent of $n,m$. We call this value $\mu_{ub}$. By Chernoff, we then have

\begin{align*}
    \sum_j \log(v^f_{ij}) \in &[(1+\frac{\epsilon}{\beta \max\{\log(\beta), 1\}}) m \mu_f, \\ &(1-\frac{\epsilon}{\beta \max\{\log(\beta), 1\}}) m \mu_f]
\end{align*}

with probability

\[\geq 1 - 2exp(-\frac{\epsilon^2}{3(\beta\\max\{\log(\beta), 1\})^2} \frac{m}{\bar c_f+ \log(\beta)} \mu_{ub})\]

which is true with high probability when $m = \Omega(\beta^2 ln^3(\beta)ln(n))$. 
\end{proof}

\subsubsection{Bounds on $\mathbf{fEMC_{aik}}$}

With this new definition, we may now prove Lemma \ref{lemma:fEMC_aik}.

\begin{proof}
We want to show that for WEFG-typical valuations

\[\sum_{k \in N_2}\sum_{a\in N_1\setminus\{i\}} fEMC_{aik} \geq -\frac{-\delta'^2 n_2}{16WC}.\]

To prove this, it will suffice to show that 

\[\forall a \in N_1\setminus\{i\}, k \in n_2, \ \Big | \frac{n_1}{w_1}\bar v_i x_a - \frac{n_2}{w_2}\bar v_i x_k \Big | \leq \frac{\delta^2}{16\beta WC}.\]

By typicality, we know that

\begin{align*}
    \sum_j \log(v^f_{ij}) \in &[(1+\frac{\epsilon}{\beta \max\{\log(\beta), 1\}}) m \mu_f, \\ &(1-\frac{\epsilon}{\beta \max\{\log(\beta), 1\}}) m \mu_f]
\end{align*}
and 
\[ \sum_j v_{aj} \in [(1-\epsilon/\beta)m\mu, (1+\epsilon/\beta)m\mu],\]

By symmetry, an identical result holds for $\bar v_{kj}^f$. Then, adding together the sources of error gives us

\begin{align*}
    \bigg |\sum_j \log(\bar v^f_{aj}) - \log(\bar v^f_{kj}) \bigg |\leq -2\frac{\epsilon}{\beta \max\{\log(\beta), 1\}} m \mu_f + \\
    m \Big( ln ((1+\epsilon/\beta)m\mu) - ln((1-\epsilon/\beta)m\mu)  \Big ).
\end{align*}

As $|\mu_f| \leq \bar c_f + \log(\beta)$, the first term may be bounded by

\[\frac{2\epsilon}{\beta}m + \frac{2\epsilon \bar c_f}{\beta }m.\]

We may bound the latter term by using the maximum value of $\frac{d}{dx} \log(x)$ over the gap. This gives

\begin{align*}
    m \Big( ln ((1+\epsilon/\beta)m\mu) - ln((1-\epsilon/\beta)m\mu)  \Big ) \\ \leq m \frac{2\epsilon  m\mu}{\beta(1-\epsilon/\beta)m\mu} \leq \frac{4\epsilon}{\beta}m.
\end{align*}

This is close to what we want, but we need a similar bound on $b_{aj}$ instead of $\bar v_{aj}^f$. We note

\begin{align*}
\bigg | \sum_j \log( b_{aj})  -\log(b_{kj}) \bigg |  &\leq \bigg | \sum_j \log( \bar v^f_{aj})  -\log(\bar v^f_{kj}) \bigg |  \\
&+ \bigg | \sum_j   \log( \bar v^f_{aj})  - \log(b_{aj}) \bigg | \\
&+ \bigg |\sum_j   \log( \bar v^f_{kj})  - \log(b_{kj})  \bigg |.
\end{align*}

We bound the latter two terms by splitting into cases. Agents $a$ and $k$ are again symmetric, so we just do it for agent $a$. If $b_{aj}$ is minimal, i.e. $b_{aj} = l/(m\beta)$, then $s_a \bar v_{aj}^f \leq l/(m\beta)$, which implies that $\bar v_{aj}^f \in [0, \frac{l}{m\beta(1-l/\beta)}]$. Applying the floor then gives

\begin{align*}
    \bar v_{aj}^f & \in \bigg [\frac{1-\epsilon/\beta}{1+\epsilon/\beta}\frac{l}{m\beta}, \frac{1}{1-l/\beta} \frac{l}{m\beta} \bigg] \\ 
    & = \bigg [\frac{1-\epsilon/\beta}{1+\epsilon/\beta}b_{aj}, \frac{1}{1-l/\beta} b_{aj} \bigg].
\end{align*}

On the other hand, if $b_{aj}$ is not minimal, then $b_{aj} = s_a \bar v^f_{aj}$, so

\[\bar v^f_{aj} \in \bigg [ b_{aj}, \frac{1}{1-l/\beta}  b_{aj}\bigg ].\]

In either case, we observe that $ |\log( \bar v^f_{aj}) - \log(b_{ij})| \leq -\log(1-\epsilon/\beta) + \log(1+\epsilon/\beta) - \log(1-l/\beta) \leq 2\frac{l}{\beta} + 3 \frac{\epsilon}{\beta}$ for $l,\epsilon \leq 1/2$. Adding all our bounds together, we get

\[ \bigg | \sum_j \log( b_{aj})  -\log(b_{kj})  \bigg |  \leq  \frac{2\epsilon \bar c_f}{\beta }m+ \frac{12\epsilon}{\beta}m   + \frac{4l}{\beta}m.\]

For our choice of $\epsilon, l$, we can see that

\[\frac{2\epsilon \bar c_f}{\beta }m+ \frac{12\epsilon}{\beta}m   + \frac{4l}{\beta}m\leq \bigg ( \frac{m}{\beta} \bigg ) \frac{\delta'^2 \mu}{32}.\]

We note that typicality implies each $\bar v_{ij} \leq \frac{2}{\mu m}$. Then,  

\begin{align*}
    \Big | \frac{n_1}{w_1}\bar v_i x_a -\frac{n_2}{w_2} \bar v_i x_k \Big |  &= \Bigg | \frac{\sum_j \bar v_{ij} \log(b_{aj}) - \bar v_{ij} \log(b_{kj})}{WC}\Bigg | \\ 
    &\leq \frac{2}{\mu m} \frac{1}{WC} \bigg ( \frac{m}{\beta} \bigg ) \frac{\delta'^2 \mu}{32} \\
    & = \frac{\delta'^2}{16 \beta WC}.
\end{align*}

Adding these bounds together across all agents $a \in N_1\setminus \{i\}, k \in N_2$ gives

\[ \sum_{k \in N_2}\sum_{a\in N_1\setminus\{i\}} \bigg (\frac{n_1}{w_1}\bar v_{i} x_{a}- \frac{n_2}{w_2}\bar v_{i} x_{k}\bigg ) \geq -\frac{-\delta'^2 n_2}{16WC}.\]

\end{proof}

\subsubsection{Bounds on $\mathbf{EMC_{aik}}$}

Finally, we prove Lemma \ref{lemma:EMC_aik}.

\begin{proof}

We want to bound deviation from $fEMC_{aik}$ via Chernoff bounds. It will be convenient to define

\[y^G_1 = \sum_{a\in N_1\setminus\{i\}} x_a,\  x_2^G = \sum_{k\in N_2} x_k\text{, and}\]
\[Y^G_1 = \sum_{a\in N_1\setminus\{i\}} X_a,\  X_2^G = \sum_{k\in N_2} X_k.\]

Note that we exclude $i$ from the sums for group 1, which is why we denote that sum with $y$ instead of $x$. Then, total deviation from the mean is bounded by the following.
\begin{align*}
    \sum_{k \in N_2}\sum_{a\in N_1\setminus\{i\}} &\bigg|\frac{n_1}{w_1}\bar v_{i} x_{a}-\frac{n_1}{w_1}\bar v_{i} X_{a}  \bigg|\\
    = n_2 &\bigg |\frac{n_1}{w_1}\bar v_i Y_1^G  - \frac{n_1}{w_1}\bar v_i y_1^G\bigg |\text{, and} \\
    \sum_{k \in N_2}\sum_{a\in N_1\setminus\{i\}} &\bigg |\frac{n_2}{w_2}\bar v_{i} x_{k}-\frac{n_2}{w_2}\bar v_{i} X_{k}  \bigg  |\\
    \leq  n_1 &\bigg |\frac{n_2}{w_2}\bar v_i X_2^G  - \frac{n_2}{w_2}\bar v_i x_2^G\bigg |.
\end{align*}

We want to apply Chernoff bounds on $\frac{n_1}{w_1} \bar v_i Y_1^G$ and $\frac{n_2}{w_2} \bar v_i X_2^G$ separately. The lower bounds are again difficult to obtain, so we define $ \Tilde{x}_a = x_a + \frac{w_1}{n_1W}$, $\Tilde{x}_k = x_k + \frac{w_2}{n_2W}$, and define $\tilde{y}^G_1, \tilde{x}^G_2$ analogously. We also recall from Section~\ref{subsec:wef} that $ \frac{n_1}{w_1}\bar v_i x_a,\frac{n_2}{w_2} \bar v_i x_k \leq \frac{2}{W} $, so

\[\frac{n_1}{w_1}\bar v_i \Tilde y^G_1 \in \bigg[ \frac{n_1-1}{W}, \frac{3(n_1-1)}{W} \bigg ] \subset  \bigg[ \frac{n_1}{2W}, \frac{3n_1}{W} \bigg ],\]

\[\frac{n_2}{w_2} \bar v_i \Tilde x_2^G \in \bigg [\frac{n_2}{W}, \frac{3n_2}{W} \bigg ].\]

where we assume $n_1 \geq 2$, as $n_1 = 1$ trivially makes every sum 0. This lets us define 

\[\Tilde X_a = X_a + \frac{w_1}{n_1W},\]
\[\Tilde X_k = X_k + \frac{w_2}{n_2W}.\]

with $\Tilde{Y}_1^G, \Tilde{X}_2^G$ defined analogously. Finally, we note that $\frac{w_g}{n_gW} \leq \frac{1}{n_g}$, so it is easy to see that 

\begin{align*}
    \Tilde Y_1^G, \Tilde X_2^G  & \in [0,2] \\
    Y_1^G - y_1^G & = \Tilde{Y}_1^G - \Tilde{y}_1^G \\
    X_2^G - x_2^G & = \Tilde{X}_2^G - \Tilde{x}_2^G \\
\end{align*}

Then, defining $\epsilon' = min\{\frac{1}{2}, \frac{\delta'^2}{96C \beta}\}$, we note that

\[ \frac{3n_1}{W} \cdot \epsilon' \leq \frac{\delta'^2n_1}{32WC\beta}.\]

Using Chernoff, we can then show

\begin{align*}
   & Pr\Big [ \Big |\frac{n_1}{w_1}\bar v_i Y_1^G  - \frac{n_1}{w_1}\bar v_i y_1^G\Big |\geq   \frac{\delta'^2n_1}{32WC\beta}  \Big  ]\\
    = & Pr\Big [ \Big |\frac{n_1}{w_1}\bar v_i \Tilde Y_1^G  - \frac{n_1}{w_1}\bar v_i \Tilde  y_1^G\Big |\geq   \frac{\delta'^2n_1}{32WC\beta} \Big  ]\\
    \leq  & Pr\Big [ \Big |\frac{n_1}{w_1}\bar v_i \Tilde Y_1^G  - \frac{n_1}{w_1}\bar v_i \Tilde  y_1^G\Big |\geq  \frac{n_1}{w_1}\bar v_i \Tilde  y_1^G \epsilon'\Big ]\\
    \leq & 2 exp(-\frac{\epsilon'^2}{3} \frac{w_1\mu m}{4n_1} \frac{n_1}{2W}) \\
    \leq & 2 exp(-\frac{\epsilon'^2}{3} \frac{w_1}{W} \frac{\mu m}{8}).
\end{align*}

Because $\epsilon' = \Omega(\frac{1}{\beta \log{\beta}})$, $m = \Omega(\beta^2  ln^2(\beta) \rho ln(n))$ suffices to ensure this is true with high probability. Moreover this implies that with high probability

\[n_2 \bigg |\frac{n_1}{w_1}\bar v_i Y_1^G  - \frac{n_1}{w_1}\bar v_i y_1^G\bigg | \leq \frac{\delta'^2n_1n_2}{32WC\beta} \leq \frac{\delta'^2n_2}{32WC}.\]

The bound on $X_2^G$ is essentially symmetric, giving us the same bound $\frac{\delta'^2n_1n_2}{32WC\beta} $, which we still reduce to $\frac{\delta'^2n_2}{32WC}$ Thus, the total error is $\leq \frac{\delta'^2n_2}{16WC}$, as desired.

\end{proof}

Finally, combining Lemmas \ref{lemma:fEMC_aik}
and \ref{lemma:EMC_aik} gives us that

\begin{align*}
    &\sum_{k\in N_2}\sum_{a\in N_1\setminus\{i\}} \bigg (\frac{n_1}{w_1}\bar v_{i} X_{a}- \frac{n_2}{w_2}\bar v_{i} X_{k}\bigg ) \\  &\geq -\frac{\delta'^2}{16WC} - \frac{\delta'^2 n_2}{16WC}
    \geq - \frac{\delta'^2 n_2}{8WC}
\end{align*}

with high probability, which combined with Lemma \ref{lemma:EMC_iik} implies that $EM_{i2} \geq 0$. Combining the worst case asymptotics from each of the three lemmas and the WEFG-typicality bound gives us our final bound of $m = \Omega( \beta^2 \log^3(\beta)\rho \log(n))$.

\subsection{Envy-Freeness with Good Types and Agent Types} \label{subsec:eft}

Finally, we also examine the problem of envy-freeness with a finite number of good types. In an instance of this problem, we have $n$ total agents split into $d$ types, with $n_a, a\in [d]$ of each type. We will assume that the agents are ordered such that $n_1\leq n_2 \leq \dots \leq n_d$. When $d = n$, the agents are fully distinct, and we show asymptotically improved bounds for this case. We also have that all goods are of one of $t$ types, and accordingly define allocations $A_{ib}$ to be the amount of good $b, b \in [t]$, received by agent $i$. Unlike our previous problems, this setting, and our algorithm, involve no probability. Accordingly, we will not pursue any guarantee of truthfulness in this setting.

For fair division with types, Gorantla et. al. \cite{Gorantla_2023} showed the following result.

\begin{restatable}{theorem}{thmgorantla} (\citet{Gorantla_2023})
    Let $1 \leq d \leq n$ and $v_1,...,v_d$ be pairwise distinct additive valuations. Let $n_1,...,n_d$ be positive integers with $\sum_{i=1}^d n_i = n$. Say there are $n_i$ agents with valuations identical to $v_i$ for all $i \in [d]$. Let $r := gcd(n_1,...,n_d)$. Then, there exists $\nu$ such that whenever $m \geq \nu \Vec{1}^t $ and $m \equiv \Vec{0}^t$ (mod r), there exists a complete EF allocation of M.
\end{restatable}

The condition that $m \equiv \Vec{0}^t$ is necessary, because agents with identical and irrational valuations can only be envy-free with identical item sets. Thus, this theorem shows that there always exists an envy-free allocation for any $n$ and sufficiently large $\nu$, unless such an allocation would be clearly impossible. Moreover, \citet{Gorantla_2023} give bounds on the number of items in two special cases. When there are exactly 2 distinct valuations, they show $\nu = O(n^2\sqrt{t}/r\delta_G)$ (where $\delta_G$ is a measure of distinctness that is different from ours), and when there are more than 2 distinct valuations but $t = 2$, they instead get the bound $\nu = O(n^2/r\delta_G)$.

Using the weighted-EF results from the previous sections, we can derive the following theorem, which gives tight bounds in $n$ and $r$ for the general problem\footnote{We treat $t$, $\delta$ as constants, and do not include them in our big-O bounds. We may alternatively write our bound as $O(\frac{n^2\log(t/\delta)}{r\delta^2})$. But, since our distance metric is different, and both distance metrics are intertwined with $t$, it is difficult to meaningfully compare these bounds to those in previous work.}, with an asymptotic improvement when agents are unique. Our mechanism can once again be implemented in polynomial time.

\begin{restatable}{theorem}{thmtypes} \label{thm:types}
    In the types setting, there exists $\nu = O(\frac{n^2}{r})$ s.t. whenever $m \geq \nu \Vec{1}^t $ and $m \equiv \Vec{0}^t$ (mod r), there exists a complete EF allocation of M. If agents are pairwise distinct, this bound improves to $\nu = O(n)$.
\end{restatable}

We use the following variation on our earlier mechanism.

\begin{enumerate}
    \item Group all agents with the same valuation $v_a$ into 1 super-agent $a$ with weight $n_a$. Split the set of items into $M_1$ and $M_2$, with $M_1$ containing $\nu$ of each item and $M_2$ containing the remainder. 
    \item Fractionally allocate $M_1$ to super-agents, using our algorithm for weighted envy-freeness (Section \ref{subsec:wef})\footnote{We will use $l = \delta/25$ and $b_{max} = \frac{2}{\nu}$, with $\delta$ defined later in this section.}. Round the resulting fractional allocation such that each super-agent gets no more than 1 extra item from each type (i.e., for each item type, round the residue for each super-agent either up or down arbitrarily across all super-agents and types). We denote super-agent $a$'s allocation of item type $b$ as $A_{ab}$.
    \item Allocate $M_2$ as evenly as possible to the original agents, such that each agent gets at most 1 extra good of each type. 
    \item Take away items from each super-agent $a$ until each $A_{ab}$ is a multiple of $n_a$. Then, take away items in multiples of $n_a$ from arbitrary super-agents, until at least $\frac{n_dn_1}{2dr}$ items of each type have been taken.
    \item Reallocate the items previously taken away, while ensuring each super-agent $a$ receives items in multiples of $n_a$. This is always possible with at least $\frac{n_dn_1}{2dr}$ items of each type. Since the total number of items taken is polynomial, this step can be done in polynomial time using a change-making algorithm.
    \item Split each super-agent into its component agents, with each agent with valuation $v_a$ receiving a $\frac{1}{n_a}$-portion of the shared bundle. 
\end{enumerate}

\begin{proof}[Proof of Theorem~\ref{thm:types}]
First, it is clear that this algorithm runs in polynomial time, as it consists of a combination of simple algorithms known to be polytime and simple operations. To prove that this algorithm returns an EF allocation, at every step we will bound the envy-margin between arbitrary agents who are not of the same type.

To do this, it helps to normalize our values slightly differently. We will now assume without loss of generality that each agent's value for $M_1$ is $1$. Accordingly, the value of $A= \Vec{1}^t$ is $\frac{1}{\nu}$ under every agent's valuation. Under this scheme, we will argue that this valuation set is still effectively typical.

First, for distinct agent types $a \neq a'$, we will define

\[\delta = min_{a \neq a'}\sum_{j \in M_1}  |v_{aj} - v_{a'j}|.\]

$\delta$ lower-bounds the total variation distance over $M_1$. Importantly, $\delta$ is scale-invariant, in the sense that if we increase $\nu$, and rescale values so that $M_1$ still has value $1$, $\delta$ remains the same. Thus, $\delta$ is some constant $>0$, which is independent of $\nu$, so we can say

\[\sum_{j \in M_1} |v_{aj} - v_{a'j}| \geq \delta.\]

Additionally, no item has value greater than $\frac{1}{\nu}$, so the upper bound never applies, and $s_i \in [1-l,1]$ by the same argument as before. This shows that in the types setting,

\[\sum_{j \in M_1} |b_{aj} - b_{a'j}| \geq \delta - 4l \geq \frac{\sqrt{2}}{2}\delta.\]

Then, again by the same arguments as before, this implies that for arbitrary agents $i,k$ in different groups, our fractional allocation over $M_1$ induces fractional envy margin

\[fEM_{ik} \geq \frac{\delta^2}{4nC}\]

for $C = - ln(l) + ln(\frac{2 m}{\nu}) =-ln(l) + ln(2t)$. Note that the original bound was expressed in $W$, but in this setting $W=n$.

For the last part of step 2, unlike before, we simply round so that each agent gets or loses at most 1 good for each type. The value of this loss/gain is bounded by $\frac{1}{\nu}$, so we see that

\[EM_{ik} \text{ after step 2} \geq  \frac{\delta^2}{4nC} - \frac{2}{\nu}.\]

Similarly, in step 3, we give away items as evenly as possible, so no agent can get more than 1 `extra' item. Thus

\[EM_{ik} \text{ after step 3} \geq  \frac{\delta^2}{4nC} - \frac{3}{\nu}.\]

Steps 4 and 5 are the trickiest, and to bound their cost, we require the following well-known result on the Frobenius coin problem.

\begin{theorem}[\citet{selmer1977linear}]
    Let $n_1 \leq n_2 \leq \dots \leq n_d$, and let $gcd(n_1,...,n_d) = 1$. Then, any $m$ larger than $ 2 n_d \lfloor \frac{n_1}{d} \rfloor - n_1$ may be partitioned into $d$ values $m_a, a \in [d]$, such that each $m_a$ is a multiple of $n_a$. 
\end{theorem}

In our case, we may simplify this bound to $\frac{2n_dn_1}{d}$. We also note that if we multiply each $n_a$ by $r$, we can still achieve this property by multiplying $m$ by $r$. As $n_dn_1$ would increase by $r^2$, we can say that in general, if $gcd(n_1,...,n_d) = r$, then $\frac{2n_dn_1}{dr}$ items suffice. 

Now, in step 4, we must first take items away from agents, so that at least $\frac{2n_dn_1}{dr}$ items are taken, and $A_{ab}$ is a multiple of $n_a$ for each $a$. To do this, we simply first take every super-agent's extra items, and then take items in multiples of $n_a$ until we reach the threshold of $\frac{2n_dn_1}{dr}$. 

Taking away the extra items again results in each agent losing at most 1 item of each type, so the cost is $\leq \frac{1}{\nu}$. To meet the threshold, in the worst case, only 1 super-agent has any goods of type $b$, and so all $\frac{2n_dn_1}{dr}$ must be taken from them. Moreover, because we take in multiples of $n_a$, it's possible that we take as much as $\frac{2n_dn_1}{dr} + n_a$ from them. As this cost is split between component agents, each agent loses at most

\[ \frac{2n_dn_1}{drn_a} + 1 \leq \frac{2n_d}{dr} + 1 \leq \frac{2n}{dr} + 1\]

of each item in order to hit the threshold. Thus, the total EM after give-away is bounded by

\[EM_{ik} \text{ after step 4} \geq  \frac{\delta^2}{4nC} - \frac{5}{\nu}-\frac{2n}{dr \nu}.\]

We also need to assign the items taken to someone. As each group had $< n_a$ extra items, we took at most $n$ of each type initially. We then took items in multiples of $n_a$ until $\frac{2n_dn_1}{dr}$. Thus, our total pool of items can be no more than $\frac{2n_dn_1}{dr} + n_a + n \leq \frac{2n_dn_1}{dr} + 2n$ of each type. 

Again, in the worst case, all the items go to the same super-agent $a$, increasing the envy of any other agent  by at most

\[\frac{1}{\nu} \bigg( \frac{2n_dn_1}{drn_a} + \frac{2n}{n_a} \bigg)\leq \frac{2n_d}{dr\nu} + \frac{2n}{r\nu} \leq  \frac{2n}{dr\nu} + \frac{2n}{r\nu}\]

where the first inequality follows from $n_a \geq r$. Thus

\[EM_{ik} \text{ after step 5} \geq  \frac{\delta^2}{4nC} - \frac{5}{\nu} - \frac{4n}{dr\nu} - \frac{2n}{r\nu}.\]

Finally, in step 6, bundles split evenly among the agents, so there is no cost. Thus, as long as

\[\frac{\delta^2}{4nC} \geq \frac{5}{\nu} + \frac{4n}{dr\nu} + \frac{2n}{r\nu}\]

each agent will be envy-free. This is true when each right side term is no more than $\delta^2/12nC$, which is true when

\[\nu \geq max\{ \frac{60nC}{\delta^2},  \frac{48n^2C}{rd\delta^2},  \frac{24n^2C}{r\delta^2}\}.\]

We may simply add all 3 terms together to obtain our final bound. Noting that $n \geq r$, this bound is asymptotically $O(\frac{n^2}{r})$ in $n$ and $r$. Although we have chosen a different constant $\delta$, we note that this bound is tight in $n$ and $r$ by the earlier work of \citet{Gorantla_2023}. We also note that when the giveaway step is unnecessary (i.e. when all agents are distinct) then it suffices to have

\[\frac{\delta^2}{4nC} \geq \frac{3}{\nu}\]

in which case $\nu > \frac{12nC}{\delta^2}$ suffices. As we clearly need $m \geq n$ for envy-freeness, this bound is asymptotically tight in $n$.
\end{proof}

\section{Conclusion and Discussion} \label{sec:conclusion}

This work studies the asymptotic fair division problem and provides an affirmative answer to a question posed by \citet{manurangsi2017asymptotic}: does there exist a truthful mechanism that finds an envy-free allocation in the asymptotic setting? We present a randomized mechanism that is truthful-in-expectation and polynomial-time implementable, and extend it to provide many positive results in other settings. Our work raises several interesting questions. Does there exist a polynomial-time mechanism that produces an envy-free allocation with high probability that is \textit{truthful ex-post}. If so, what bounds do we require on the number of items? Can we additionally achieve Pareto optimality alongside truthfulness and envy-freeness? For the complementary setting of asymptotic fair division of chores, it is known that $\Theta(n)$ items are sufficient for envy-freeness (\citet{manurangsi2025chores}) when $n$ is large. Can one extend this bound to the design of truthful mechanisms in the chores setting?

\bibliography{ref}
\clearpage

\end{document}